\documentclass[format=acmsmall, review=false]{acmart}
\usepackage{acm-ec-26}
\usepackage{booktabs} 
\usepackage[ruled]{algorithm2e} 

\SetAlFnt{\small}
\SetAlCapFnt{\small}
\SetAlCapNameFnt{\small}
\SetAlCapHSkip{0pt}
\IncMargin{-\parindent}

\setcitestyle{authoryear}

\newtheorem{theorem}{\sc Theorem}[section]  
\newtheorem*{theorem*}{\sc Theorem}
\newtheorem{example} [theorem]{\sc Example}  
\newtheorem{definition}  [theorem]{\sc Definition}
\newtheorem{lemma}[theorem]{\sc Lemma}  
\newtheorem{claim}[theorem]{\sc Claim}

\newtheorem{proposition}[theorem]{\sc Proposition}
\newtheorem{remark}{\sc Remark}
\newtheorem{assumption}{\sc Assumption}

\newcommand{\Hos}{H}
\newcommand{\Doc}{S}
\newcommand{\Com}{K}
\newcommand{\Match}{\mathcal{M}}

\usepackage{amsmath}
\usepackage{amsthm}
\usepackage{graphicx}
\usepackage{comment}
\usepackage{xcolor}
\usepackage{url}
\usepackage{tikz,pgfplots}
\usetikzlibrary{shapes.geometric, arrows}
\usepackage{thm-restate}
\usepackage{natbib}

\newcommand{\Xomit}[1]{}

\newcommand{\randbud}[1]{\mathcal{B}_{#1}}
\newcommand{\z}{{\bf z}}
\newcommand{\y}{{\bf y}}
\newcommand{\x}{{\bf x}}
\newcommand{\prices}{{\bf p}}

\SetKwInput{KwData}{Input}
\SetKwInput{KwResult}{Output}

\def\final{0}  
\def\iflong{\iffalse}
\ifnum\final=0  
\newcommand{\thanh}[1]{{\color{blue}[{\small Thanh: \bf #1}]\marginpar{\color{red}*}}}
\newcommand{\ysl}[1]{{\color{brown}[{\small Young-San: \bf #1}]\marginpar{\color{red}*}}}
\newcommand{\haoyu}[1]{{\color{olive}[{\small Haoyu: \bf #1}]\marginpar{\color{red}*}}}
\else 
\newcommand{\thanh}[1]{}
\newcommand{\ysl}[1]{}
\newcommand{\haoyu}[1]{}
\fi  

\title[Matching with Committee Preferences]{Matching with Committee Preferences}

\author{Haoyu Song, Th\`anh Nguyen , Young-san Lin}

\begin{abstract}
We study a many-to-one matching model inspired by school choice, where schools evaluate applicants using multiple rankings rather than a single priority order. We model each school’s evaluation with social choice criteria to reflect the school’s internal ranking process. In particular, we define acceptable choices as candidates ranked above a top percentile of the accepted cohort by a sufficient number of evaluators. Stability is then defined in terms of acceptability: accepted candidates must receive strong support, while rejected candidates receive at most weak support. Since exact acceptability and stability may not exist, we construct approximately stable outcomes using a new equilibrium concept that combines matching with a Lindahl equilibrium over ordinal preferences, providing a flexible, equilibrium-based framework for committee-based matching markets.
\end{abstract}

\begin{document}

\begin{titlepage}

\maketitle

\vspace{1cm}
\setcounter{tocdepth}{2} 
\tableofcontents

\end{titlepage}


\section{Introduction}

In many-to-one matching markets, such as school choice and residency matching, institutions may report a single ranking of applicants, but it often conceals underlying disagreements. Decisions reflect multiple evaluations, based on distinct criteria and the differing judgments of committee members. In university admissions, applicants are assessed on exam scores, research potential, recommendation letters, and diversity considerations. Hospitals assembling residency cohorts weigh interview performance, clinical evaluations, staffing needs, and institutional priorities. Hiring decisions similarly integrate assessments that reflect varying views on technical ability, cultural fit, and compliance requirements.

A standard response to multiple rankings is to collapse them into a single \emph{weak} order and resolve remaining indifferences through tie-breaking rules \citep{abdulkadiroglu2009strategy, erdil2008matter}, or to address conflicts by imposing reserves for different applicant types \citep{abdulkadiroglu2003school,hafalir2013effective}.
 Both approaches effectively impose an external resolution of ranking disagreements. In contrast, we internalize the multiplicity of rankings directly into each school’s choice, defining acceptance and rejection through social choice–theoretic criteria. This approach aligns closely with classical matching theory: the Gale–Shapley framework \citep{gale_shapley_1962} and its extension via choice functions \citep{hatfield_milgrom_2005} evaluate candidates relative to the \emph{admitted cohort}. Acceptance requires a candidate to have sufficient support to displace an admitted peer, while rejection follows from the absence of such support. By framing choices this way, our method naturally accommodates multi-dimensional rankings while preserving the core intuition of cohort-based stability.

However, extending this logic to committee-based preferences is not straightforward. In the classical Gale–Shapley setting, rejection is justified by comparing an unmatched candidate to the worst admitted candidate. With multiple committee rankings, such a comparison is neither natural nor unique: different rankings may identify different “worst” candidates. Rejection therefore involves two decisions—choosing a benchmark within each ranking and determining how much overall support a candidate needs across rankings to displace an admitted peer.

We introduce a parameterized $(\vec{\alpha},\beta)$ framework to formalize how committees aggregate individual judgments. 
Consider a school $h$ with capacity $c_h$ admitting candidates $C^*$ from the applicant pool $S$, where non-wastefulness requires $|C^*| = \min\{c_h, |S|\}$. 
When applicants exceed capacity, a committee determines preferences: each member $k \in K_h$ provides an ordinal ranking $\succ_k$, and the parameters $(\vec{\alpha},\beta)$ translate these rankings into a collective decision that balances support across the committee. 
For each member $k$, $\alpha_k \in [1, c_h]$ sets the benchmark for evaluating candidates relative to the admitted cohort: $\alpha_k = 1$ compares to the top-ranked member of $C^*$, $\alpha_k = c_h$ to the lowest, and intermediate values correspond to positions in between. 
Member $k$ \emph{supports} a candidate if the candidate ranks at least as high as this benchmark, and a candidate’s overall support is the number of members who support them.\footnote{Weighted support or other scoring rules can be incorporated by modifying the framework; see Remark~\ref{remark:1}.}

Given parameters $(\vec{\alpha},\beta)$, a set of admitted candidates $C^*$ at a single school $h$ is \emph{acceptable} if every admitted candidate has support at least $\beta$, while every rejected candidate has support at most $\beta$—that is, each admitted candidate has weakly higher support than any rejected candidate when using the accepted cohort as the reference. 
This notion extends naturally to multiple schools: a matching is \emph{stable} if, at each school, the support of every admitted student is at least as high as that of any student assigned to a less-preferred school.

A key distinction of our framework is that, because each committee member uses the accepted cohort as the reference for benchmarking, the acceptable choice is \emph{endogenous} even for a single school: a student’s acceptance depends on the benchmark for each ranking, which in turn depends on the set of accepted students. Endogenous benchmarks capture interactions across committee members: acceptance standards depend on who else is accepted rather than on fixed thresholds, reflecting the iterative, deliberative nature of committee decisions.\footnote{Related ideas appear in economic theory, including choice with endogenous reference points, reference-dependent preferences \cite{koszegi2006model, masatlioglu2005reference}, and salience-driven evaluations \cite{bordalo2013salience}.}

Importantly, this endogeneity allows the framework to capture a broad spectrum of choice concepts, connecting classical social choice notions such as the Condorcet winning set with standard matching-theoretic properties like responsiveness and path independence. The generality, however, comes at a cost: for a single institution, finding an acceptable set reduces to a multiwinner problem aggregating committee preferences, which can be empty for certain parameters. Furthermore, even when single-institution choices are well-defined, they do not necessarily satisfy substitutes, so a stable matching may fail to exist at the market level.

Despite impossibility results, our main results provide constructive guarantees: we can produce solutions that are \emph{nearly} acceptable and stable. 
We show that an acceptable choice exists if each $\alpha_k$ is perturbed by roughly the size of the committee, and achieving near stability may similarly require adjusting the school’s capacity by a comparable magnitude.  

This is feasible in practice because, in many applications, the number of committee members is small relative to the number of available positions, making the necessary perturbations minimal in percentile terms. 
For example, with 10 committee members and 1000 positions, adjusting each $\alpha_k$ by the committee size alters the benchmark percentile by around 1\%.  Such minor adjustments preserve the original committee benchmarks while ensuring that the resulting matching is nearly stable.

\subsection*{Related Work}

Our approach differs in that it respects schools’ internal rankings in their actual selection decisions. In standard models, schools’ priorities are often represented as \emph{weak orders}—allowing ties—within the deferred acceptance framework. Ties naturally arise when schools evaluate applicants on multiple criteria, and a common solution is to resolve them using a uniform lottery \citep{abdulkadiroglu2009strategy}. Although strategy-proof, such lotteries treat all tied applicants equally, overlooking subtle differences in schools’ internal evaluations and potentially reducing fairness or efficiency. Alternative tie-breaking methods \citep{erdil2008matter, ashlagi2019assigning} similarly do not account for these internal rankings. Another line of work addresses ties through reserves or quotas \citep{abdulkadiroglu2003school, hafalir2013effective, ehlers2014school, sonmez2022affirmative}, which avoid tie-breaking but require \emph{ex ante} capacity allocations based on prior information, which can be difficult to justify.


Our framework sidesteps these issues by directly incorporating schools’ multiple criteria into the selection process. Prior research along these lines generally falls into two strands.  The first highlights \emph{impossibility and computational hardness}: multi-modal preferences may admit no stable matching, and determining whether a stable matching exists can be NP-hard \citep{chen2018stable, chen2021fractional, wen2022position}. The second provides \emph{positive algorithmic or fairness results} under restrictive assumptions; for example, \cite{boehmer2025proportional} links matching to multiwinner social choice with dichotomous preferences, while \cite{school_japan} applies a weaker M-fairness notion. By contrast, we allow general rankings and recover stability in an approximate sense, yielding solutions that remain meaningful even when exact stability fails.

For single-school choice, our approach relates closely to the large literature on multiwinner voting, which generally distinguishes between ranking-based and approval-based preferences. Ranking-based rules, such as the Chamberlin–Courant \citep{chamberlin1983representative} and Monroe \cite{monroe1995fully} rules, leverage full preference orders to maximize voter satisfaction. Approval-based rules, including Proportional Approval Voting (PAV) and rules satisfying extended justified representation (EJR) \citep{aziz2017justified}, allow voters to indicate which candidates they find acceptable, emphasizing proportionality and fairness. Ranking rules preserve rich ordinal information, while approval rules provide clear thresholds for representation, highlighting a trade-off between preference granularity and representation.

Our notion of acceptability extends Condorcet-winning sets \citep{elkind2015condorcet}  by combining the strengths of both approaches. Unlike classical Condorcet sets, which compare rejected candidates only to each voter’s top accepted candidate, our method compares each rejected candidate against every voter’s top percentile of candidates, which serves as that voter’s approval set. Furthermore, our concept ensures that accepted candidates receive strong approval, while rejected candidates receive at most weak approval. Prior work, including \cite{elkind2015condorcet,CharikarLassotaRamakrishnanVettaWang2025,voting_us_26}, guarantees only weak approval for rejected candidates, making our approach a stricter and more informative generalization.

From a technical perspective, our approach departs from classical methods in matching theory. We adopt an \emph{equilibrium relaxation} followed by \emph{rounding} \citep{matching2018, nguyen2023complementarities}: we first use an equilibrium concept to capture stability and acceptability in a fractional solution, and then round this solution to an integral allocation that nearly preserves these properties. A key distinction from prior work \citep{matching2018} and followed up works \citep{allocation2021, nguyen2021stability} is that our setting combines social choice  with stability. This requires a new form of \emph{equilibrium relaxation}, which incorporates additional interactions and constraints related to public goods, rather than focusing solely on private goods as in previous frameworks.

In particular, for the single-school choice problem, our work builds on the Lindahl equilibrium with ordinal preferences (LEO) \citep{nguyen-song} and its extensions in social choice \citep{voting_us_26}.
We adopt a similar equilibrium relaxation but employ a different and more delicate rounding procedure—based on iterative rounding \citep{kiraly2012degree}—to guarantee stronger structural properties than in prior work.
Specifically, our approach ensures not only that rejected candidates receive weak support, but also that accepted candidates receive sufficiently strong support, a requirement that is not imposed in these earlier papers.

For multiple schools, we extend this approach to stable matching by introducing a new equilibrium concept that combines LEO with stability in a single framework.
By treating each school–student pair as a public good, the equilibrium captures preferences on both sides of the market.
We call this concept the Matching Equilibrium with Ordinal Preferences (MEO), which provides a unified way to analyze multi-winner allocation across multiple institutions.

\section{Model}

Our model applies to a wide range of applications, including school choice and resident matching, among others. For expositional clarity, we adopt the terminology of school choice  throughout the paper.

Let $\Hos$ denote the set of $m$ schools and $\Doc$ denote the set of $n$ students (we sometimes use the term ``candidates” or ``applicants" interchangeably). Each student $i \in \Doc$ has a strict preference ordering $\succ_i$ over $\Hos \cup \emptyset$ , and each school $h \in \Hos$ has a capacity $c_h$ and a selection committee $\Com_h$, where each member $k \in \Com_h$ has a strict preference ordering $\succ_k$ over $\Doc\cup \emptyset$. Denote $K = \bigsqcup_h K_h$, the disjoint union of the sets $K_h$. Here we assume that students strictly prefer attending a school to remaining unmatched, and that schools strictly prefer filling a seat to leaving capacity unfilled. That is, the outside option $\emptyset$ is the least preferred option in each ranking.\footnote{This assumption can be relaxed by introducing ``dummy'' students or ``dummy'' schools representing these outside options.}


We use $\Match$ to denote the set of feasible matching between students and schools, where each student is only matched/admitted to one school and a school $h$ is matched to a set of at most $c_h$ students. Given a feasible matching $M\in \Match$ and a school $h$, we denote the set of students admitted by a school $h$ as $M_h$, because of feasibility, we have $|M_h|\leq c_h$. Given a student $i$, we denote the school to which $i$ is admitted as $M_i$. $H_i=\emptyset$ means that student $i$ is unmatched. We abuse the notation and say that $(i,h)\in M$ when $i\in M_h$ and $h = M_i$.

\subsection{Acceptable Choice Set}

We begin by formulating the selection problem from the perspective of a single school. 
Unlike the classical matching literature, decisions are made by a committee rather than a single evaluator. Committee members evaluate candidates along different criteria, leading to heterogeneous and potentially conflicting rankings. When multiple candidates must be chosen, the problem is not to aggregate these rankings into a single order, but to determine a set of candidates that is collectively defensible. This is inherently a social choice and multi-winner problem: disagreement is unavoidable, yet the outcome must withstand challenges from alternative candidates under the committee’s diverse preferences. Our notion of acceptable choice set is designed to capture this requirement.

Fix a school $h$ with capacity $c_h$ and selection committee $K_h$. Each committee member $k \in K_h$ is associated with an integral \textit{rank parameter} $\alpha_k\in \{1,2,...,c_h\}$ and we use $\vec{\alpha}$ to denote the corresponding \textit{rank vector} of the school.

Given a set of students $C^* \subseteq S$,  we use $C^*_{\alpha_k}$ to denote the $\alpha_k$-th favorite student in $C^*$ and set $C^*_{\alpha_k} = \emptyset$ if $|C^*| < \alpha_k$. We will call $C^*_{\alpha_k}$ the \textit{$\alpha_k$-rank student} of $C^*$ according to the committee member  $k\in K_h$.


Now, we say that a committee member $k$ \emph{$\alpha_k$-approves} a student $i$
with respect to $\vec{\alpha}$ a $C^*$ if
$i \succeq_k C^*_{\alpha_k}.$ When clear from the context, we will drop $\alpha_k$ and $C^*$ and say $k$ approves a student $i$.

Next, the degree of support for $i$ at school $h$ is measured by the number  of committee members who 
approves $i$ with respect to $\vec{\alpha}$ and $C^*$:
\[
\mathrm{Support}^{\Vec{\alpha}}_h(i,C^*) := \bigl| \{ k \in \Com_h : i\succeq_k C^*_{\alpha_k} \} \bigr|.
\]

\begin{remark}\label{remark:1}

In our definition of support, we assume that each committee member evaluates candidates in a binary manner—approving those above her rank threshold—and that support is the number of approving members. These assumptions can be relaxed. Member importance can be incorporated by assigning weights and defining support as the total weight of approving members. More generally, monotonic step-wise scoring rules can be accommodated by introducing multiple copies of a member with different thresholds and weights.\footnote{For example, a percentile-based scoring rule can assign 2 points to candidates in the top 25\%, 1 point to those in the 25–50\% range, and 0 points otherwise; this can be implemented by introducing two copies of the member with thresholds at the 25th and 50th percentiles.} For ease of presentation, we focus on this class of support measures.
\end{remark}

Using the support functions, we define our main solution concept for the single-school choice problem, called \emph{acceptable choice set}.

\begin{definition}[Acceptable choice set] \label{def:acs}
    Fix a school $h$ with capacity $c_h$,  approval thresholds $\vec{\alpha}$, $\beta$, and a set of candidates $C \subseteq S$. A set $C^* \subseteq C$ is a \emph{acceptable choice set} with respect to  the set of applicants $C$ if the following holds: 
    \begin{enumerate} 
    \item Non-wastefulness: $|C^*| = \min(c_h,|C|)$; 
    \item Individual Rationality: $ \mathrm{Support}^{\vec{\alpha}}_h(i, C^*)
    \;\ge\; \beta   \text{ for every } i \in C^*$; 
    \item No blocking: $\mathrm{Support}^{\vec{\alpha}}_h(j, C^*)\le \beta  \text{ for every } j \in C \setminus C^*. $
    \end{enumerate}
\end{definition}

The non-wastefulness condition requires a school to fill its available positions by selecting as many candidates as possible, up to its capacity. The individual rationality and no-blocking conditions use a threshold parameter $\beta$ to determine which candidates are accepted or rejected. Specifically, candidates with support strictly above $\beta$ are accepted, those with support strictly below $\beta$ are rejected, and candidates with support exactly equal to $\beta$ lie on the boundary. In some cases, tie-breaking among boundary candidates is necessary, as a complete separation between accepted and rejected candidates may not be possible. We illustrate this  with the following example.

\begin{example} 
Consider two rankings over three candidates: 
\[
a \succ_1 b \succ_1 c \quad (\alpha_1 = 1), \qquad a \succ_2 c \succ_2 b \quad (\alpha_2 = 1),
\] 
and suppose the school has capacity 2. That is, we must select two candidates from \(\{a,b,c\}\).

The set \(\{b,c\}\) is not acceptable, as candidate 1 has strictly higher support. The acceptable choice sets are \(\{a,b\}\) or \(\{a,c\}\), with \(\beta = 0\). This example highlights that  tie-breaking is required: due to the capacity constraint, candidates with the same level of support can be both accepted and rejected. 
\end{example}

The next examples illustrate that the acceptable choice set framework can capture well-known concepts in social choice, such as the Condorcet winner or winning set. In these examples, the no-blocking condition based on $\beta$ provides a quantitative explanation for rejected candidates: a candidate is rejected if they do not receive sufficient support. In general, $\beta$ should be chosen as small as possible, and our results establish an asymptotically optimal bound.

\begin{example}[Condorcet winner and Condorcet winning set]\label{ex:3}
Consider a school~$h$ with capacity~$c_h$ and a set~$K$ of committee members. 
Suppose that $\alpha_k = 1$ for all $k \in K$. 
For a given parameter~$\beta$, an \emph{acceptable choice} satisfies two conditions:
\begin{enumerate}
\item The school selects $c_h$ candidates such that no rejected candidate is preferred to the best candidate in the selected set by more than $\beta$ committee members. 
When $\beta = \tfrac{1}{2}|K|$ and $c_h = 1$, this requirement coincides with the definition of a Condorcet winner. 
For general $c_h$, it corresponds to a Condorcet winning set, which need not exist when $c_h \le 3$.
\item In addition, every selected candidate must receive at least $\beta$ support relative to the selected set: each accepted candidate is ranked highest among the selected candidates by at least $\beta$ committee members.
\end{enumerate}
\end{example}

\begin{example}[$(t,\gamma)$-undominated set]Given a committee~$K$, \citet{voting_us_26} study a generalization of the Condorcet winning set by requiring that, for each unselected candidate, the fraction of committee members who prefer that candidate to the $t$-th most preferred selected candidate is at most~$\gamma$. This condition corresponds exactly to the no-blocking requirement of an acceptable choice set with parameters $\alpha_k = t$ for all $k \in K$ and $\beta = \gamma |K|$.

Additionally, such an acceptable choice set requires that every selected candidate ranks among the top $t$ of the selected cohort for at least $\beta = \gamma |K|$ committee members.
\end{example}

Thus, our notion of acceptability strengthens the concept of a $(t,\gamma)$-undominated set by additionally requiring that accepted candidates receive sufficiently high support, rather than focusing solely on limiting the support of rejected candidates. 
Nevertheless, even without this additional condition, the lower bound from Theorem 3  of \citet{voting_us_26} applies directly to our setting.

\begin{proposition}[\citep{voting_us_26}]\label{prop:lowerbound}
Given a committee~$K$, a number of candidates to be selected $c_h$, and $\alpha_k = t \le c_h$ for all $k \in K$, there exist instances in which no acceptable choice set exists whenever
\[
\beta < |K| \cdot \frac{t+1}{c_h+1}.
\] 
\end{proposition}
This result implies that we cannot guarantee the existence of an acceptable set in which every unselected candidate receives less support than the $t$-th best accepted candidate by more than a $\frac{t+1}{c_h+1}$ fraction of the committee.

\Xomit{
We actually have a stronger result. 
\begin{proposition}
There exists a school $h$ with a capacity $c_h$, a committee $K$ and $\alpha_k=t\leq c_h$ for all $k\in K$ s.t. 
\begin{enumerate}
\item for any $\beta<|K|\cdot \frac{t+1}{c_h+1}$, there does not exist any set of students satisfying no blocking; 
\item for any $\beta\geq 0$, there does not exist any $\beta$-acceptable choice set;
\end{enumerate}
\end{proposition}
}

\subsection{Stable Matching}
We extend the single-organization choice set framework to a matching environment.

First, we demonstrate that varying $\vec{\alpha}$ and $\beta$ allows the acceptable choice-set framework to recover many standard choice functions in the matching literature.


\begin{example}[Responsive Choice] 
This is the classical choice function in the framework of \cite{gale_shapley_1962}. Each school $h$ has a single ranking over applicants, and given a set of applicants $C$, the school selects the top $c_h$ candidates in $C$. 

Suppose $|C| > c_h$ and the school $h$ has a single ranking $\succ_h$. Set $\alpha = c_h$ so that the benchmark candidate is the least-preferred among the accepted ones, and let $ \beta =0$. Because there is only one ranking, the support for each candidate is either $1$ or $0$. Therefore, a $\beta$-acceptable choice set must include all candidates with support $1$. In this case, the only possible $\beta$-acceptable choice set coincides exactly with the responsive choice.
\end{example}

\begin{example}[Path-Independent Choice Function]
Path independence defines a well-known, general class of choice functions that is sufficient to guarantee stability
\citep{path_independence_81, choice_and_matching}.   
One characterization of a path-independent choice function is via decomposition into a finite set of committee rankings $\{\succ_k\}_{k \in K}$ over the students. Given a set of applicants $C \subseteq S$, the choice from $C$ is
\(
r(C) = \bigcup_{k \in K} \{\text{Top}_k(C)\},
\)
where $\text{Top}_k(C)$ denotes the most preferred candidate in $C$ according to ranking $\succ_k$.

The individual rationality and no blocking condition (without the non-wastefulness) of acceptable choice set can capture path independence by setting the benchmark candidate for each ranking to be the top choice, i.e., $\alpha_k = 1$ for all $k \in K$, and choosing $\beta=1/2$ (any $\beta \in (0,1)$ works).  

Under this specification, any rejected candidate has support from no committee member, meaning they are ranked below the top candidate in every ranking. Conversely, each accepted candidate has support from at least one committee member, i.e., they are the top choice in at least one ranking. This construction precisely yields the choice $r(C)$ defined above.
\end{example}

However, due to the generality of our framework, many “non-substitutable” choice functions can be captured. The following example illustrates this.
\begin{example}
Consider 5 committee members with rankings over candidates $a, b, c, d$, denoted by $\succ_k$ for member $k$:
$$
 d \succ_{1,2} a \succ_{1,2} b \succ_{1,2} c ;  \;\; 
 d \succ_{3,4} b \succ_{3,4} a \succ_{3,4} c ;  \;\; 
c \succ_5 d \succ_5 a \succ_5 b.
$$

Suppose we need to choose 2 candidates, with thresholds $\beta = 1$ and $\alpha = 1$ for all rankings. The support of a candidate is determined by the top-ranked candidate among the selected set in each ranking. 

\begin{enumerate}
    \item For the candidate set $\{a, b, c\}$:  It is easy to check that the only acceptable set is $\{a, b\}$. (Under this set, the support counts for $a,b$ and $c$ are $3$, $2$, and $1$, respectively.) 
    \item For the candidate set $\{a, b, c, d\}$:  
    The only acceptable set is $\{d, c\}$. This is because candidate $d$ beats both $a$ and $b$ in all rankings, so $d$ must be included. Once $d$ is included, the support for $a$ and $b$ drops to 0, while $c$ has support 1 because it beats $d$ in member 5's ranking.
\end{enumerate}

This example illustrates that the choice function is not substitutable: without $d$, candidate $c$ is rejected, but when $d$ is present, $c$ becomes selected. There is a complementarity between $c$ and $d$.
\end{example}

\Xomit{
\begin{definition}[$\beta$ stable matching]
Given parameters $\{\beta_h\}_{h\in H}$ and $\{\alpha_k\}_{k\in K}$, a feasible matching $M$ is \emph{stable} if, for every school $h$, 
 the set of students currently matched to $h$, $M_h$, is \emph{$\beta_h$-acceptable} with respect to the  set of applicants $M_h \cup \{\, i : h \succ_i M_i \,\}$.

\end{definition}
}

To define stability, we start with the notion of blocking.

\begin{definition}[Blocking Pair]
Given parameters $\{\beta_h\}_{h \in H}$ and $\{\alpha_k\}_{k \in K}$, and a feasible matching $M$, a student-school pair $(j,h)$ is a \emph{blocking pair} if
\begin{enumerate}
    \item the student strictly prefers the school to her current assignment, that is $h \succ_j M_j$, and
    \item the school either has unfilled capacity, $|M_h| < c_h$, or finds the student sufficiently acceptable, i.e.,
    \[
    \mathrm{Support}^{\vec{\alpha}}_h(j, M_h) > \beta_h .
    \]
\end{enumerate}
\end{definition}

\begin{definition}[Stable Matching]
Given parameters $\{\beta_h\}_{h \in H}$ and $\{\alpha_k\}_{k \in K}$, a feasible matching $M$ is \emph{stable} if it satisfies both of the following conditions:
\begin{enumerate}
    \item \emph{Individual rationality}: every student $i$ assigned to school $h$ is acceptable to that school, i.e.,
    \[
    \mathrm{Support}^{\vec{\alpha}}_h(i, M_h) \ge \beta_h 
    \quad \text{for all } i \in M_h;
    \]
    \item \emph{No blocking}: the matching admits no blocking pair.
\end{enumerate}
\end{definition}
A stable matching guarantees that each school admits students with high support, while any student preferring the school but not admitted has weakly lower support.

As in the standard literature, an acceptable choice set for a single school can be interpreted as a special case of a stable matching. Specifically, it corresponds to a two-school market in which one school represents the actual institution and the other serves as an outside option that all students rank below the first.  Thus, our framework provides a general model encompassing both social choice and stable matching. The impossibility result in Proposition~\ref{prop:lowerbound} therefore implies the nonexistence of an acceptable choice set—and, equivalently, of a stable matching in this market—unless further assumptions are imposed on the parameters $\vec{\alpha}$ and $\beta$.

The impossibility result is driven by environments in which the number of rankings (or committee members) is large relative to the school’s capacity. This contrasts with typical school-choice applications, where the number of rankings or evaluation criteria is small compared to capacity—for example, on the order of tens of criteria versus thousands of available seats.

Motivated by this distinction, we focus on environments in which the number of rankings or evaluation criteria is relatively small compared to capacity. For example, in college admissions, the number of criteria used to evaluate applicants is typically on the order of ten, whereas institutional capacity is often in the thousands.



\section{Approximate Acceptable Choice Set}
We begin with a single-school problem and develop an algorithm to construct an approximately acceptable choice set for a single school. This will serve as a building block for our results in the next section on multiple schools.

The main result of this section is the following. 
\begin{theorem}\label{thm:main_social_choice}
Consider an instance of the single-school matching problem $(c_h, \vec{\alpha})$ for a school $h$ with a committee $K_h$, where each $k \in K_h$ has strict preferences over a set of students $S$.  
Then, there exists a subset of candidates $C^* \subseteq S$ that is acceptable with respect to the parameters $\{\alpha'_k\}_{k \in K_h}, \beta$  such that
$$\beta\leq  \Bigl\lceil \sum_{k\in K_h}\frac{\alpha_k}{c_h} \Bigr\rceil \;\; \text{ and }  \;\; |\alpha_k-\alpha'_k|\leq 2 \beta \quad \forall k \in K_h.$$
\end{theorem}

\begin{remark}
To interpret this result, it is helpful to view $\alpha_k$ in relative rather than absolute terms. 
Specifically, each committee member $k$ can be thought of as approving a candidate by comparing her to the top $\alpha_k / c_h$ fraction of the ranking. 

Consider the following example. 
Suppose the committee size is $|K_h| = 10$ and the capacity is $c_h = 1000$. 
If each committee member compares candidates to the top $20\%$ of the ranking, ($\alpha_k / c_h = 0.20$, so $\alpha_k = 200$). 
Our bound implies that $|\alpha'_k - \alpha_k| \le 2 |K_h| \cdot \alpha_k / c_h = 4$. 
In percentile terms, this corresponds to a deviation of at most $4/1000 = 0.4\%$.

Notice, the bound on $|\alpha'_k - \alpha_k|$ is at most twice the committee size. 
In percentile terms, this corresponds to a deviation of at most $2|K_h|/c_h$, which is typically small when capacity is large relative to the committee.
However, this bound can be substantially tighter when the percentile level $\alpha_k / c_h$ are small. 
In particular, if $\alpha_k / c_h < 1/(2|K_h|)$, the deviation is at most 1.

A second remark concerns $\beta$. 
The upper bound on $\beta$ guarantees that no rejected candidate can be supported by more than 
$\sum_{k \in K_h} \alpha_k / c_h$ committee members. 
Interpreting $\alpha_k / c_h$ in percentile terms, this means that if all committee members compare candidates to the top $p$ fraction of their rankings, then no rejected candidate can receive support from more than a $p$ fraction of the committee. 

This establishes a tight relationship between the comparison threshold in committee preferences and the maximal support that a rejected candidate can obtain. 
Such bounds are known to be asymptotically tight in voting environments (see \cite{voting_us_26}, and Proposition~\ref{prop:lowerbound} above), and here we extend them to a more general setting in which committee members may use heterogeneous percentile thresholds.
\footnote{To clarify this bound, notice that Proposition~\ref{prop:lowerbound} shows that if $\alpha_k = \alpha$, then for
\(
\beta = |K| \frac{\alpha+1}{c_h+1} - \epsilon,
\)
there exist instances where no acceptable choice set exists. In contrast, Theorem~\ref{thm:main_social_choice} guarantees that, for any instance, an (approximate) acceptable choice set exists for some
\(
\beta \le \Big\lceil |K| \frac{\alpha}{c_h} \Big\rceil.
\)
When $\alpha$ and $c_h$ scale at the same rate and are large, our bound approaches the impossibility threshold, showing that it is asymptotically tight. Together, these results provide a nearly sharp characterization of the range of feasible $\beta$.}

\end{remark}

\subsection*{Overview}
As illustrated earlier, the notion of an acceptable choice both generalizes and strengthens several classical social choice problems, although an acceptable choice need not always exist. We therefore introduce the concept of an \emph{acceptable fractional set}, which mirrors the defining properties of an acceptable set but allows for fractional solutions.

We then show that whenever an acceptable fractional set exists, it can be converted into an integral solution via an iterative rounding algorithm. The resulting outcome is an acceptable set after only small adjustments to the ranking parameters $\vec{\alpha}$ and the capacity $c_h$. The key feature of the algorithm is that it preserves a strong separation between accepted and rejected candidates: accepted candidates enjoy strong support, while rejected candidates have weak support. This yields acceptability guarantees that are strictly stronger than those obtained by prior approaches.

Finally, to establish the existence of an acceptable fractional set, we adopt an equilibrium-based approach. Specifically, we adapt the Lindahl equilibrium with ordinal preferences introduced by \cite{nguyen-song}, and extend it to accommodate heterogeneous agent weights $\alpha_k$.

\subsection*{Acceptable Fractional Set}

We start by extending the notions of $\alpha_k$-rank student and support to the fractional setting, which coincide with the old definitions when the underlying allocation is integral.


\begin{definition}[$\alpha_k$-rank student]\label{def:boundary_student}
Let $S$ be a set of $n$ students, $h$ be a school with capacity $c_h \in \mathbb{N}$, and 
$\vec{\alpha} \in [0,c_h]^{|K_h|}$ be a vector of ranking thresholds.  
Fix a committee member $k \in K_h$ with preference ordering $\succ_k$ over $S$.  
Given a fractional assignment $\x \in [0,1]^n$, the \emph{$\alpha_k$-rank student} $i_{k,\x}^{\vec{\alpha}}$ for $k$ is the most-preferred student who reaches the $\alpha_k$ threshold in cumulative assignment among students ranked at least as high.  
If the total assigned mass is less than the school capacity, we set $i_{k,\x}^{\vec{\alpha}} = \emptyset$. Formally,
\[
i_{k,\x}^{\vec{\alpha}} :=
\begin{cases}
\emptyset, & \text{if } \sum_{i\in S} x_i < c_h, \\
\max_{\succ_k}\left\{ i \in S \;\middle|\; 
\sum_{i' \succeq_k i} x_{i'} \ge \alpha_k \right\}, & \text{otherwise.}
\end{cases}
\]
\end{definition}

\begin{definition}[Fractional Approval]
Let $S$ be a set of students, $h$ be a school, and $\vec{\alpha} \in [0,c_h]^{|K_h|}$. Fix $k \in K_h$ with preference $\succ_k$ and a fractional assignment $\x \in [0,1]^n$.

Committee member $k$ \emph{weakly $\alpha_k$-approves} student $i$ if $i \succeq_k i_{k,\x}^{\vec{\alpha}}$, and \emph{strongly  $\alpha_k$-approves} $i$ if $i \succ_k i_{k,\x}^{\vec{\alpha}}$.  
The \emph{weak (strong) support} for $i$, denoted $\mathrm{w\text{-}Support}_h^{\vec{\alpha}}(i,\x)$ (resp., $\mathrm{s\text{-}Support}_h^{\vec{\alpha}}(i,\x)$), is the number of committee members who weakly (strongly) $\alpha_k$-approve $i$.
\end{definition}

\begin{figure}[htbp]
    \centering
    \includegraphics[width=0.45\linewidth]{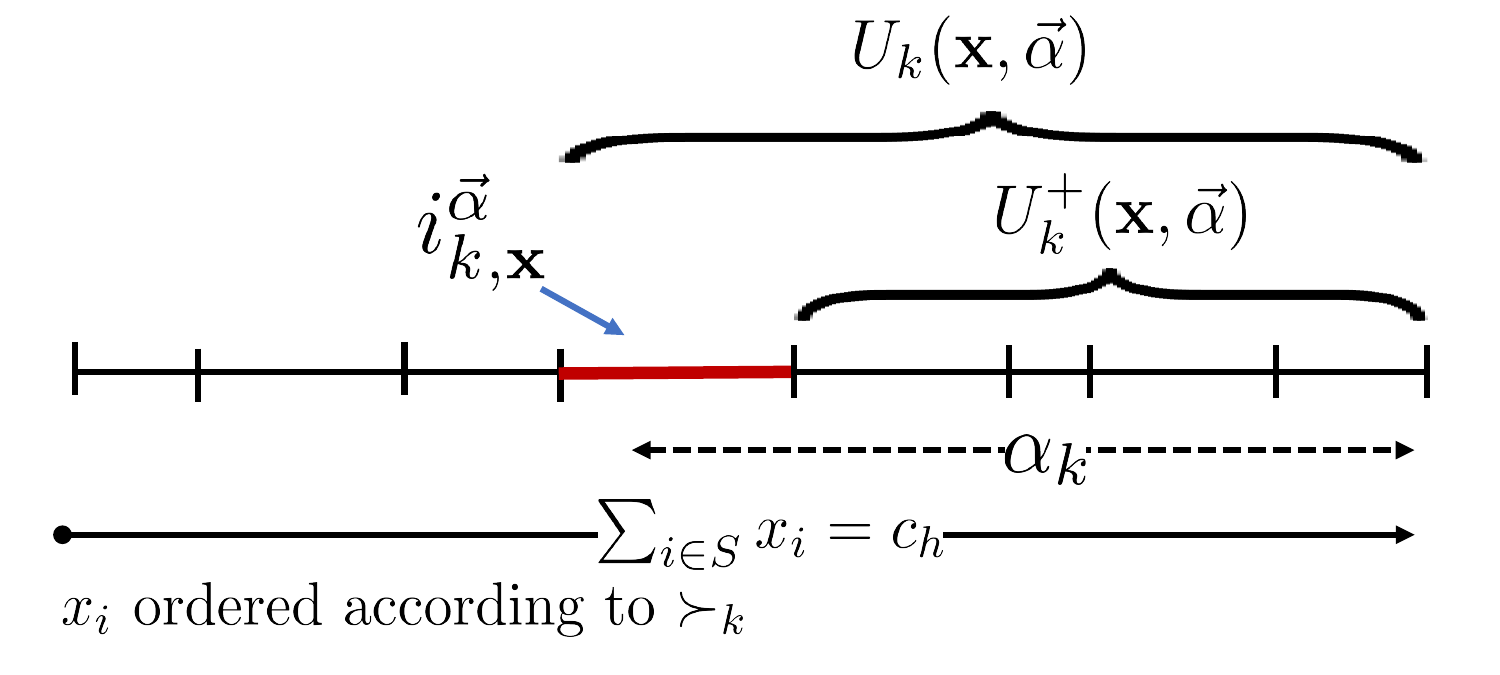}
    \caption{$\alpha_k$-ranked student}
    \label{fig:01}
\end{figure}
Figure~\ref{fig:01} provides a visualization of an $\alpha_k$-ranked student and the support measures $\mathrm{w\text{-}Support}_h^{\vec{\alpha}}(i,\x)$ and $\mathrm{s\text{-}Support}_h^{\vec{\alpha}}(i,\x)$. Assume that $\x$ is a fractional vector, define
\begin{equation}\label{eq:D}
\begin{aligned}
U_{k}(\x, \vec{\alpha})
&:= \{\, i \in S : i \succeq_k i_{k,\x}^{\vec{\alpha}} \,\}, \\ 
U^+_{k}(\x, \vec{\alpha})
&:= \{\, i \in S : i \succ_k i_{k,\x}^{\vec{\alpha}} \,\},
\end{aligned}
\end{equation}
i.e., the set of students that are weakly (and strongly) preferred by $k$ over the $\alpha_k$-rank student.
We have
\begin{equation}\label{eq:D1}
\begin{aligned}
\mathrm{w\text{-}Support}_h^{\vec{\alpha}}(i,\x) &= \bigl|\{k\in K_h \mid i\in U_k(\x, \vec{\alpha})\}\bigr|, \\
\mathrm{s\text{-}Support}_h^{\vec{\alpha}}(i,\x) &= \bigl|\{k\in K_h \mid i\in U_k^{+}(\x, \vec{\alpha})\}\bigr|.
\end{aligned}
\end{equation}
Moreover, by the definition of the $\alpha_k$-rank student, $i_{k,\x}^{\vec{\alpha}}$, we have 
\begin{equation}\label{eq:D2}
\sum_{i \in U_{k}(\x, \vec{\alpha})} x_i \ge \alpha_k  \ge \sum_{i \in U^+_{k}(\x, \vec{\alpha})} x_i \ge \alpha_k-1.    
\end{equation}
The first two inequalities of \eqref{eq:D2} hold because, by definition, 
\(i_{k,\x}^{\vec{\alpha}}\) is the most preferred candidate such that the total mass of candidates ranked at least as high reaches \(\alpha_k\).
The last inequality of \eqref{eq:D2} follows from the fact that \(x_{i_{k,\x}^{\vec{\alpha}}} \le 1\).

With these notions of supports, we define fractional acceptable set.


\begin{definition}[Acceptable fractional set]\label{def:frac_stable}
Let $S$ be a set of $n$ students and $h$ be a school with capacity $c_h \in \mathbb{N}, c_h\le n$, with parameters $\vec{\alpha}$ and  $\beta$. A vector $\x \in [0,1]^n$ is an \emph{acceptable fractional set} if the following conditions hold:
\begin{enumerate}
\item (Non-wastefulness) $\lVert \x \rVert_1 = c_h$; \label{frac-condition1}
\item (Individual rationality) for every student $i$ with $x_i > 0$,
\(
\mathrm{w\text{-}Support}_h^{\vec{\alpha}}(i,\x) \ge \beta; \;\;\; 
\) \label{frac-condition2} 
\item (No blocking) for every student $i$ with $x_i < 1$,
\(
\mathrm{s\text{-}Support}_h^{\vec{\alpha}}(i,\x) \le \beta.
\) \label{frac-condition3}
\end{enumerate}
\end{definition}

\Xomit{
With the new notion of support, we can also define fractional acceptable set.

\begin{definition}[acceptable Fractional set]\label{def:frac_stable}
Let a set of $n$ students $S$, a school $h$ with a capacity $c_h\in \mathbb{N}$ and a rank vector $\Vec{\alpha}\in [0,c_h]^{|K_h|}$, and a committee member $k\in K_h$  be given. Then, given $\beta\in[0,1]$, a vector $\z\in [0,1]^n$ is an acceptable fractional set if the following conditions hold: 
\begin{enumerate}
\item \textbf{Non-Wasterfulness:} $|\z|_1 = c_h$; 
\item \textbf{Individual Rationality:} for each $i$ with $z_i>0$, $\mathrm{w\text{-}Support}_h^{\vec{\alpha}}(i,\z)\geq\beta$; 
\item \textbf{No Blocking:} for each $i$ with $z_i<1$, $\mathrm{s\text{-}Support}_h^{\vec{\alpha}}(i,\z)\leq \beta$. 
\end{enumerate}
\end{definition}
}

\subsection*{Proof of Theorem~\ref{thm:main_social_choice}: Rounding Fractional Sets}

We prove Theorem~\ref{thm:main_social_choice} by rounding an acceptable fractional set.

The following two lemmas provide the key steps for constructing an acceptable integral set from a fractional one. All proofs are deferred to Appendix~\ref{app:rounding}.

\begin{lemma}\label{lemma:round}
Let $\z \in [0,1]^n$ be an acceptable fractional set with respect to $\vec{\alpha} \in [0,c_h]^{|K_h|}$ and $\beta$. Suppose we round $\z$ to an integral vector $\z' \in \{0,1\}^n$ such that $\lVert \z' \rVert_1 = \lVert \z \rVert_1 = c_h$ and $z'_i = z_i$ whenever $z_i$ is integral ($0$ or $1$).  
For each $k$, let  
\[
\alpha'_k := \sum_{i \in U_k(\z, \vec{\alpha})} z'_i.
\]  
Then the set $C^* := \{i : z'_i = 1\}$ is an acceptable set with respect to $\vec{\alpha}'$.
\end{lemma}

Lemma \ref{lemma:round} shows that if there exists a rounding algorithm satisfying certain properties, an acceptable fractional set can be rounded to an integral one with adjusted parameters $\alpha$. The next Lemma shows the existence of such a rounding algorithm and quantifies the magnitude of the change of parameters. 

\begin{lemma}\label{lemma:iterround}
Let $\z \in [0,1]^n$ be an acceptable fractional set with respect to $\vec{\alpha} \in [0,c_h]^{|K_h|}$ and $\beta$. 
Then, there exists an algorithm rounding $\z$ to $\z' \in \{0,1\}^n$ as in Lemma~\ref{lemma:round} such that
\[
|\alpha'_k - \alpha_k| \le 2 \beta  \quad \forall k \in K_h.
\]
\end{lemma}
In the rounding procedure, we formulate a linear program with two types of constraints. 
The first enforces feasibility by requiring that the total mass is $c_h$. 
The second preserves acceptability by ensuring that, for each committee member $k$, the total mass assigned to $U^+_k(\z,\vec{\alpha})$ is maintained. The key observation is that, because $\z$ is an acceptable fractional set, any variable $z_i<1$ has strong support at most $\beta$. 
Consequently, such a variable belongs to at most $\beta$ sets $U^+_k(\z,\vec{\alpha})$. 
This bounded-degree property allows us to apply the iterative rounding technique of \citet{kiraly2012degree}, yielding an integral solution with only small additive violations.  The full proof is provided in the appendix.

To complete the proof of Theorem~\ref{thm:main_social_choice}, it remains to show the existence of an acceptable fractional set with the desired parameters. 
This result is summarized in the following lemma.

\begin{lemma}\label{thm:1stable_fractional}
Let $h$ be a school with capacity $c_h$, committee members $K_h$, a ranking vector 
$\vec{\alpha}\in [0,c_h]^{|K_h|}$, and set of applicants $S$. 
There exists $\beta^*\le  \lceil \frac{\sum_{k\in K_h}\alpha_k}{c_h} \rceil $ and an acceptable fractional set $\z$ with respect to $\vec{\alpha}, \beta^*$. 
\end{lemma}

\subsection*{Proof of Lemma~\ref{thm:1stable_fractional}}
Finally, we establish existence of an acceptable fractional set.

The proof uses  a generalized Lindahl equilibrium with ordinal preferences, an adaptation of $\alpha$-LEO \citep{voting_us_26,nguyen-song}. For completeness, we introduce the concept.  

A Lindahl equilibrium consists of agents (committee members), public goods (applicants), personalized prices for each agent–alternative pair, and a common outcome. At equilibrium, each agent's optimal allocation given her personalized prices coincides with the common outcome, which also maximizes aggregate revenue.

Extending Lindahl equilibria to ordinal preferences requires adjustments. Following \cite{nguyen-song}, the Lindahl equilibrium with ordinal preferences (LEO) uses continuous budgets and probabilistic outcomes.  In particular, under a fixed budget $b \in \mathbb{R}_+$, a committee member $k$'s demand, which is her most-preferred affordable applicant,  can be discontinuous in the prices. To handle this, each committee member $k$ is assigned a random budget according to a distribution $\randbud{}$ supported on $[0,1]$, and the corresponding random demand under a personalized price vector $\prices_k$ is
\[
\mathcal{D}_k(\prices_k, \randbud{}) := \Big\{ \max_{\succ_k} \{ i \in S \cup \{\emptyset\} : p_{k,i} \le b \} \;\big|\; b \sim \randbud{} \Big\}.
\]
This demand defines a lottery over the choices, where each outcome corresponds to a realization of the budget, and under each realization, the agent selects her most-preferred affordable option. A key property of random demand is that it is continuous in the prices. Using this property, \cite{nguyen-song} established the existence of a LEO, defined as follows (adapted to our setting).


Consider an economy in which students are treated as ``public goods''. There are two types of agents: committee members, who express preferences over students, and a central agent for the school  who decides the final allocation of students while taking the committee members’ preferences into account.

Each committee member $k$ is associated  with a personalized price vector $\prices_k \in [0,1]^{n+1}$ over students, with the price of the empty allocation set to zero: $\prices_{k,\emptyset} = 0$.

The central agent representing school $h$, denoted also by $h$, selects a fractional allocation of students that maximizes the total revenue according to the personalized prices of the committee members, subject to the school's capacity $c_h$:
\[
\mathcal{D}_h(\{\prices_k\}_{k\in K_h}) = \arg\max_{\bold{x} \in [0,1]^n} \sum_{i\in S} \Big(\sum_{k\in K_h} p_{k,i}\Big) \cdot x_i \quad \text{s.t.} \quad \sum_{i \in S} x_i = c_h.
\]

Combining the committee members’ demands with the central agent’s allocation, we can formally define a LEO, which coordinates individual preferences and the school’s capacity constraints.

\begin{definition}\label{def:LEO}
Let $(c_h, \vec{\alpha} := \{\alpha_k\}_{k\in K_h})$ be an instance of the single-school problem, and let a random budget $\randbud{}$ distribution supported on $[0,1]$ be given. A  Lindahl equilibrium with ordinal preferences (LEO) is a tuple $(\y, \z, \prices)$ consisting of:
\begin{itemize}
    \item a personalized price vector $\prices_k \in [0,1]^n$ for each committee member $k \in K_h$, with $p_{k,\emptyset} = 0$,
    \item an individual demand $\y_k \in [0,1]^{n+1}$ over students for each committee member $k \in K_h$, and
    \item a central allocation $\z \in [0,1]^n$ over students.
\end{itemize}
The tuple $(\y, \z, \prices)$ satisfies:
\begin{enumerate}
    \item $y_{k,i} = \Pr(\mathcal{D}_k(\prices_k, \randbud{}) = i)$ for each committee member $k \in K_h$ and student $i \in S \cup \{\emptyset\}$; \label{LEO-condition1}
    \item $\z \in \mathcal{D}_h(\{\prices_k\}_{k \in K_h})$ for the school $h$; \label{LEO-condition2}
    \item $\alpha_k \cdot y_{k,i} \le z_i$ for each committee member $k \in K_h$ and student $i \in S$, with strict inequality only if $p_{k,i} = 0$. \label{LEO-condition3}
\end{enumerate}
\end{definition}
Compared with the traditional Lindahl equilibrium, which enforces equality, 
condition~\eqref{LEO-condition3} requires only a weak inequality between each committee member’s demand and the centrally chosen allocation, scaled by \(\alpha_k\).  Strict inequality is permitted only when the corresponding price is zero. 
This relaxation, together with the continuity of random demand, enables a fixed-point argument to establish the existence of a LEO (see Appendix for the proof).

\begin{theorem}[\cite{nguyen-song}]\label{thm:LEO_existence}
Let a random budget $\randbud{}$ supported on $[0,1]$, a school $h$ with capacity $c_h$, and a ranking vector $\vec{\alpha}$ be given. If the cumulative distribution function of $\randbud{}$ is continuous, then a LEO exists.
\end{theorem}

Using LEO, we construct an acceptable fractional lottery with the desired parameters. 
This construction is summarized in the following result, which directly proves Lemma~\ref{thm:1stable_fractional}.

\begin{theorem}\label{thm:exist}
Let $h$ be a school with capacity $c_h$, committee members $K_h$, a ranking vector $\vec{\alpha} \in [0,c_h]^{|K_h|}$, and a set of applicants $S$. Then there exists $\varepsilon \in (0,1)$ and $\beta^*\leq \lceil \sum_{k\in K_h}\frac{\alpha_k}{c_h} \rceil$ such that when $\randbud{}$ is a uniform distribution on $[1-\varepsilon, 1]$, the common allocation $\z$ of the corresponding LEO $(\y, \z, \prices)$ is an acceptable fractional set with respect to $(\vec{\alpha},\beta^*)$.
\end{theorem}

We outline the main ideas of the proof below; the complete proof appears in the Appendix~\ref{app:LEO-accept}.
The non-wastefulness of $\z$ follows directly from the definition of LEO.
Next, since each agent receives an allocation whose price does not exceed her budget, we can bound each agent’s expected payment. By condition~\eqref{LEO-condition3} of LEO, the total expected revenue collected by the centralized agent under $\z$ is at most $\sum_k \alpha_k$. This, in turn, allows us to bound the total price of each candidate that is selected with positive probability under $\z$ (Claim~\ref{clm:prices_LEO}).
Finally, observe that for each committee member $k$, the $\alpha_k$-ranked candidate is exactly the most preferred candidate whose price is at most $1-\varepsilon$ (Claim~\ref{lemma:boundary_student}). This yields an lower  bound on the price of any candidate strictly preferred to the $\alpha_k$-ranked candidate and, consequently, allows us to bound the number of agents who find another candidate more attractive (Claim~\ref{clm:LEO_approx}).

\section{Multiple Schools: A Matching Problem}

In this section, we extend the analysis to setting consists of multiple schools. 
Our main result is the following.
\begin{theorem}[Existence of Approximately Stable Matching]\label{thm:mainmatch}
Let $(\Hos, S, K = \bigcup_{h \in \Hos} K_h)$ be a matching market, where each school $h$ has capacity $c_h$ and each committee member $k \in K_h$ has parameter $\alpha_k$.   Then there exist adjusted capacities $\{c'_h\}_{h \in \Hos}$, adjusted parameters $\{\alpha'_k\}_{k \in K}$, and thresholds $\{\beta_h\}_{h \in \Hos}$ such that
\[
|c'_h - c_h| \le 2|K_h| + 1, \quad
|\alpha_k - \alpha'_k| \le  2|K_h| + 2, \quad \text{and} \quad
\beta_h \le \Biggl\lceil \sum_{k \in K_h} \frac{\alpha_k}{c_h} \Biggr\rceil,
\quad \text{for all } h \in \Hos, k \in K_h.
\]
Under these adjustments, there exists a matching $M$ that is stable with respect to $(\vec{c}', \vec{\alpha}', \vec{\beta})$.
\end{theorem}

Thus, similar to Theorem~\ref{thm:main_social_choice}, we can perturb each school’s parameters by at most twice the size of its committee to guarantee the existence of a stable matching. As illustrated above, when the committee size is small relative to the school’s capacity, these adjustments are minimal: only a small number of additional seats may be needed, and the thresholds for each committee member are altered by a negligible percentage.

The high-level idea parallels the notion of an acceptable set from the previous section. We first compute a (fractional) matching using an equilibrium concept, and then round it to an integral solution. Section~\ref{sec:meo} introduces the equilibrium concept and establishes its existence (Theorem~\ref{thm:meo_existence}). Section~\ref{sec:frac} shows that this equilibrium corresponds to a fractional stable matching (Theorem~\ref{thm:matching_exists}), and Section~\ref{sec:round} describes how to convert the fractional matching into an integral stable matching (Lemmas~\ref{lemma:round_matching} and \ref{lemma:round_matching_violation}).

The main challenge in the multi-school setting is that preferences must be respected on both sides.  
Unlike the single-school case, where each student could be treated as a public good for committee members, here students have preferences over schools, and each school’s committee members have preferences over students.  
To over come this, we model each school–student pair $(h,i)$ as a public good, allowing both sides to express preferences within a unified equilibrium framework.   We further introduce the school as an agent coordinating the matching.   Its utility aggregates the prices of both committee members and students for each pair $(h,i)$, capturing the combined evaluation of the school and the student.  
This yields the equilibrium concept we call \emph{MEO} (Matching Equilibrium with Ordinal Preferences).

Before proceeding to the technical details, we make the following simplifying assumption.

\begin{assumption} 
We assume that the total number of students is at least the sum of the total capacities of schools, i.e.,
\(
n = |S| \ge \sum_{h \in H} c_h.
\)
\end{assumption}
This assumption is without loss of generality. As is standard in the literature, we can introduce dummy students so that, for each school $h$, the relative ranking of the original students in $S$ remains unchanged, while all dummy students are ranked below the least-preferred original student and ordered arbitrarily among themselves. The presence of dummy students does not affect the existence or properties of stable matchings, and hence stability is preserved with or without them.

\subsection{MEO: Matching Equilibrium with Ordinal Preferences}\label{sec:meo}
Given a set of students $S$, a set of schools $\Hos$, and for each school $h \in \Hos$ a committee $K_h$, we consider an economy with three types of actors: a central agent representing each school $h$, the committee members in $K=\bigcup_{h\in H} K_h$—each of whom has preferences over the students in $S$—and students, who have preferences over the set of schools $\Hos$.

As in the single-school setting, each committee member $k \in K_h$ is endowed with a personalized price vector 
$\prices_k \in [\delta,1]^{n+1}$ over students, with $p_{k,\emptyset} = \delta$.\footnote{The parameter $\delta$ is introduced for technical reasons, which will become clear in the next subsection.}
Given a random budget according to a distribution $\randbud{}$, supported on $[0,1]$ and identical across committee members, agent $k$ induces a random demand 
$\mathcal{D}_k(\prices_k,\randbud{})$ according to her preference ordering $\succ_k$.

A key difference from the single-school model is that students are no longer treated as passive ``products,'' but instead act as decision-makers.
Accordingly, each student $i$ is associated with a personalized price vector 
$\boldsymbol{q}_i \in [0,1]^{m+1}$ over schools, with $q_{i,\emptyset} = 0$.
Student $i$ is also endowed with the budget drawn from $\randbud{}$ and induces a random demand 
$\mathcal{D}_i(\boldsymbol{q}_i,\randbud{})$ based on her preference ordering $\succ_i$.\footnote{Agents may in general have different budget distributions; for our purposes, it suffices to use a common distribution.}

Another difference is that the central agent $h$ no longer maximizes total revenue.
Instead, $h$ is endowed with a utility $u_{h,i}$ for each student $i$, defined as the product of the committee members’ cumulative prices for $i$ and student $i$’s price for $h$, namely,
\[
u_{h,i}
\;:=\;
q_{i,h} \cdot \sum_{k \in K_h} p_{k,i}.
\]
The central agent $h$ then chooses an allocation of ``seats'' to students so as to maximize total utility, subject to the school’s capacity constraint:
\[
\mathcal{D}_h\big(\{\prices_k\}_{k\in K_h}, \{\boldsymbol{q}_i\}_{i\in S}\big)
\;: = \;
\arg\max_{w \in [0,1]^{S}}
\sum_{i\in S} u_{h,i}
\, w_i
\quad
\text{s.t.}
\quad
\sum_{i\in S} w_i = c_h .
\]
The intuition behind $u_{h,i}$ is that prices assigned by students and committee members serve as proxies for how highly they rank a given pairing.
Accordingly, a hospital benefits from a match with student $i$ only when both the committee members and the student assign a high value to that pairing.
The multiplicative structure of $u_{h,i}$ also plays a crucial role: as we show later, it ensures that—even though each school $h$ acts independently—at equilibrium the total assignment of any student $i$ across all schools in $\Hos$ is at most one.

\begin{definition}[Matching Equilibrium with Ordinal Preferences (MEO)]\label{def:MEO}
Let $(\Hos,S,K=\bigcup_{h\in\Hos}K_h)$ be a matching market with capacities
$\{c_h\}_{h\in\Hos}$ and parameters $\vec{\alpha}=\{\alpha_k\}_{k\in K}$.
Fix a random budget distribution $\randbud{}$ supported on $[0,1]$ and $\delta\in(0,1]$.

A \emph{matching equilibrium with ordinal preferences} (MEO) is a tuple
$(\x,\y,\z,\prices,\bold{q})$ where
$\x_i\in[0,1]^{m+1}$, $\y_k\in[0,1]^{n+1}$, $\z_h\in[0,1]^{n+1}$,
$\bold{q}_i\in[0,1]^{m+1}$ with $q_{i,\emptyset}=0$,
and $\prices_k\in[\delta,1]^{n+1}$ with $p_{k,\emptyset}=\delta$,
such that:
\begin{enumerate}
    \item $x_{i,h}=\Pr(\mathcal{D}_i(\bold{q}_i,\randbud{})=h)$
    for all $i\in S$, $h\in\Hos\cup\{\emptyset\}$; \label{meo-condition1}
    \item $y_{k,i}=\Pr(\mathcal{D}_k(\prices_k,\randbud{})=i)$
    for all $k\in K_h$, $i\in S\cup\{\emptyset\}$; \label{meo-condition2}
    \item $\z_h\in\mathcal{D}_h(\{\prices_k\}_{k\in K_h},\{\bold{q}_i\}_{i\in S})$
    for all $h\in\Hos$; \label{meo-condition3}
    \item $x_{i,h}\le z_{h,i}$, with strict inequality only if $q_{i,h}=0$;\label{meo-condition4}
    \item $\alpha_k y_{k,i}\le z_{h,i}$, with strict inequality only if $p_{k,i}=\delta$. \label{meo-condition5}
\end{enumerate}
\end{definition}

Here, $\x$, $\y$, and $\z$ denote the demands of students, committee members, and schools, respectively, with conditions~(\ref{meo-condition1})--(\ref{meo-condition3}) simply restating that each agent induces a random demand and each school allocates seats to maximize total utility.

Conditions~(\ref{meo-condition4}) and (\ref{meo-condition5}), inspired by competitive equilibrium, coordinate student and committee member demands with the schools’ central allocations.  
Here, $x_{i,h}$ is the fractional seat demanded by student $i$ at school $h$, $y_{k,i}$ the fractional seat demanded by committee member $k$ for student $i$, and $z_{h,i}$ the fractional seat allocated to $i$ by school $h$.  
Condition~(\ref{meo-condition4}) ensures that a school’s allocation meets or exceeds the student’s demand, with strict excess only if the student’s price is zero, while condition~(\ref{meo-condition5}) similarly coordinates the allocation with the committee member’s demand, scaled by~$\alpha_k$.

Next, we show that a MEO exists.  
The proof follows similarly to that of Theorem~\ref{thm:LEO_existence} and is therefore deferred to the appendix.  

\begin{theorem}\label{thm:meo_existence}
Let $(\Hos, S,K, \{c_h\}_{h\in \Hos}, \{\alpha_k\}_{k\in K})$ be an instance of the matching problem.  
Let $\delta>0$ and let $\randbud{}$ be a random budget supported on $[0,1]$.  
If the cumulative distribution function of $\randbud{}$ is continuous, then a matching equilibrium with ordinal preferences (MEO) exists.
\end{theorem}

\subsection{Stable Fractional Matchings from MEO}\label{sec:frac}
In this section, we show that the central allocations of a MEO correspond to a stable fractional matching, which we define shortly.

Before establishing stability, we first show that the central allocations in a MEO form a \emph{valid} fractional matching, in the sense that no student is assigned more than one unit of total capacity. More precisely:

\begin{definition}[Fractional Matching]
Let $S$ and $\Hos$ be the sets of students and schools, with each $h \in \Hos$ having capacity $c_h$.  
A vector $\z = \{\z_h\}_{h \in \Hos}$ with $\z_h \in [0,1]^{|S|}$ is a \emph{fractional matching} if
\[
\sum_{i \in S} z_{h,i} \le c_h \ \forall h \in \Hos, 
\quad \text{and} \quad
\sum_{h \in \Hos} z_{h,i} \le 1 \ \forall i \in S.
\]
\end{definition}

\begin{lemma}\label{lemma:equal_x_z}
Let $(\Hos, S, K=\bigcup_{h\in \Hos} K_h)$ with capacities $\{c_h\}_{h\in \Hos}$ and ranking parameters $\vec{\alpha} := \{\alpha_k\}_{k \in K}$ be a matching instance.  
Let $(\x, \y, \z, \prices, \bold{q})$ be a MEO under a random budget distribution $\randbud{}$ supported on $[0,1]$ and $\delta>0$.  
If $|S| \ge \sum_{h \in \Hos} c_h$, then
\(
z_{h,i} = x_{i,h} \quad \text{for all } i \in S,\, h \in \Hos.
\)  
Consequently, $\z := \{\z_h\}_{h \in \Hos}$ constitutes a fractional matching.
\end{lemma}

\begin{proof}
Suppose, for contradiction, that there exists a student $i^*$ and school $h^*$ with $z_{h^*,i^*} > x_{i^*,h^*}$, which implies $z_{h^*,i^*} > 0$.  
By the MEO definition, this requires $q_{i^*,h^*} = 0$, hence $u_{h^*,i^*} = 0$.  

Since $|S| \ge \sum_{h \in \Hos} c_h$, there exists a student $i'$ not fully assigned to a school, so $x_{i',\emptyset} > 0$.  
Because $\x_{i'}$ represents student $i'$’s random demand and $\emptyset$ is the least-preferred option, $x_{i',\emptyset} > 0$ implies $q_{i',h^*} > 0$; otherwise, $i'$ would consume $h^*$ instead of $\emptyset$.  

Moreover, $p_{k,i'} \ge \delta > 0$ for all $k \in K_{h^*}$, so $u_{h^*,i'} > 0$.  
Thus, the central agent $h^*$ can strictly increase total utility by reallocating mass from $i^*$ to $i'$, contradicting the optimality of $\z_{h^*}$.
\end{proof}

Next, we define the notion of stability for fractional matchings. For notational convenience, we introduce the concept of a student’s \emph{boundary school} with respect to a given fractional matching.

\begin{definition}[Boundary School]
Let $\Hos$ be a set of schools, $i$ a student, and $\z = \{\z_h\}_{h \in \Hos}$ a fractional matching.  
The boundary school of student $i$ is the least-preferred school of $i$ that receives a strictly positive assignment, i.e.,
\[
h_{i,\z} :=
\begin{cases}
\min_{\succ_i} \{h \in \Hos : z_{h,i} > 0\}, & \text{if } \sum_{h \in \Hos} z_{h,i} = 1,\\
\emptyset, & \text{otherwise.}
\end{cases}
\]
\end{definition}

With this notation in place, a stable fractional matching is defined as follows.

\begin{definition}[Stable Fractional Matching] 
Let $\Hos$ be a set of schools and $S$ a set of students.  
A fractional matching $\z = \{\z_h\}_{h \in \Hos}$ with $|\z_h|_1 = c_h$ for each $h$ is \emph{stable} with respect to parameters $\{\beta_h\}_{h \in \Hos}$ and $\{\alpha_k\}_{k \in K}$ if it satisfies both of the following conditions:

\begin{enumerate}
    \item \textbf{Individual rationality:}  
    Every student $i$ assigned positively to school $h$ receives sufficient weak support:
    \[
    \text{w-}\mathrm{Support}^{\vec{\alpha}}_h(i, \z) \ge \beta_h. 
    \]

    \item \textbf{No blocking:}  
    For any student $i$ not fully assigned to school $h$ but assigned positively to a less-preferred school, i.e., 
    $z_{i,h} < 1$ and $h \succ_i h_{i,\z}$, the strong support from $h$ is small:
    \[
    \text{s-}\mathrm{Support}^{\vec{\alpha}}_h(i, \z) \le \beta_h.
    \]
\end{enumerate}
\end{definition}

\Xomit{
\begin{definition}[Fractional Blocking Pair]\label{def:fractional_blocking} Let a set of schools $\Hos$, a set of students $S$ and a fractional matching $\z:\{\z_h\}_{h\in \Hos}$ be given. Then, with parameters $\{\beta_h\}_{h \in H}$ and $\{\alpha_k\}_{k \in K}$, a student--school pair $(j,h)$ is a \emph{blocking pair} with respect to $\z$ if
\begin{enumerate}
    \item the student is not fully assigned to  $h$ i.e. $z_{h,i}<1$, and \label{fractional-blocking-condition1}
    \item the student strictly prefers the school to her boundary school, that is $h \succ_j h_{j,\z}$, \label{fractional-blocking-condition2} and
    \item the school either has unfilled capacity, $|\z_h|_1 < c_h$, or finds the student having a sufficient amount of \textit{strong} support with respect to $\z$, i.e.,
    \[
    \text{s-}\mathrm{Support}^{\vec{\alpha}}_h(j,\z) > \beta_h .
    \] \label{fractional-blocking-condition3}
\end{enumerate}

\end{definition}

As seen in Definition \ref{def:fractional_blocking}, condition (\ref{fractional-blocking-condition1}) gives the pre-requisite for a blocking pair: the pair has to be \textit{not} fully matched in the first place. Then, as in condition (\ref{fractional-blocking-condition2}), the student decides whether to deviate with her boundary school as a benchmark. Lastly, the school decides whether to deviate after checking whether the student has received a sufficient amount of support within the school, as seen in condition (\ref{fractional-blocking-condition3}).
}

We next show that, given a fixed ranking vector $\vec{\alpha}$, a MEO with a random budget distribution $\randbud{}$ being the uniform distribution on $[1-\varepsilon,1]$ can be used to construct a stable fractional matching with the desired $\beta$-factors.   

\begin{theorem}\label{thm:matching_exists}
Let $(\Hos, S, \{c_h\}_{h\in \Hos}, \vec{\alpha} := \{\alpha_k\}_{k\in K})$ be a matching instance with $\sum_{h\in \Hos} c_h \le |S|$.  
Then there exist $\varepsilon, \delta > 0$ such that the fractional matching $\z := \{\z_h\}_{h\in \Hos}$ derived from the corresponding MEO is stable with respect to $\vec{\alpha}$ and parameters $\{\beta_h\}_{h\in \Hos}$ satisfying
\[
\beta_h \le \Big\lceil \frac{\sum_{k \in K_h} \alpha_k}{c_h} \Big\rceil \quad \text{for all } h \in \Hos.
\]
\end{theorem}
The proof follows the same roadmap as in the single-school case, with one important difference: we reason in terms of utility levels  since each school now maximizes utility instead of revenue (see Appendix~\ref{app:thm:matching_exists}). 

The key idea is to relate a student’s incentive to deviate to a school to the level of support she receives from that school’s committee. 
First, we show that if a student~$i$ could find a better  school~$h$ compared with  one of the schools to which she is assigned with positive probability, then the corresponding price~$q_{i,h}$ must be at least $1-\varepsilon$. 
However, the fact that student~$i$ is not fully assigned to~$h$ in equilibrium implies that admitting her would not generate sufficiently high utility for the school. 
In particular, the utility
\(
u_{h,i} = q_{i,h} \cdot \sum_{k \in K_h} p_{k,i}
\)
cannot be large. 
This, in turn, implies that the total price $\sum_{k \in K_h} p_{k,i}$, must be small. 
Translating prices back into support yields an upper bound on student~$i$’s support at~$h$, which rules out profitable deviations and establishes stability.

\subsection{Rounding}\label{sec:round}

Similar to the single-school case, we give two lemmas used for constructing an integral stable matching from a fractional one.

As a counterpart to Lemma \ref{lemma:round}, the following Lemma gives the prerequisites for a rounding procedure in order to convert a fractional stable matching to an integral one.  See Appendix~\ref{app:roundmatch} for the proof.

\begin{lemma}\label{lemma:round_matching}
Let $\{\z_h\}_{h \in \Hos}$ be a stable fractional matching with respect to $\{\alpha_k\}_{k \in K}$ and $\{\beta_h\}_{h \in \Hos}$.  
Suppose we round it to binary vectors $\{\z'_h\}_{h \in \Hos}$ such that 
\begin{itemize}
    \item $z'_{h,i} = z_{h,i}$ if $z_{h,i} \in \{0,1\}$, and
    \item $\sum_{h \in \Hos} z'_{h,i} = 1$ if $\sum_{h \in \Hos} z_{h,i} = 1$.
\end{itemize}
Let $M$ be the matching with $M_h = \{i \in S : z'_{h,i} = 1\}$, $c'_h = |M_h|$, and $\alpha'_k = \sum_{i \in D_k(\z_h,\vec{\alpha})} z'_i$.
Then $M$ is stable with respect to $\{c'_h\}_{h \in \Hos}$, $\{\alpha'_k\}_{k \in K}$, and $\{\beta_h\}_{h \in \Hos}$.
\end{lemma}

This lemma shows that a fractional stable matching can be rounded to an integral matching while preserving stability as follows.  
First, any entries that are already integral remain unchanged.  
Second, if a student is fully assigned to a school in the fractional solution, she remains fully assigned in the integral solution.  

The resulting integral matching may differ from the fractional one in two parameters:  
(1) the total capacity at each school, $c'_h$, and  
(2) the upper-set totals for each committee member, $\alpha'_k := \sum_{i \in D_k(\z_h, \vec{\alpha})} z'_i$, where $D_k(\z_h, \vec{\alpha})$ is the set of students weakly or strongly preferred by $k$ over her $\alpha_k$-rank student (see Figure~\ref{fig:01}).

We now show that there exists a rounding procedure satisfying the conditions given in Lemma \ref{lemma:round_matching} at the expense of only a small change of parameters. 
\begin{lemma}\label{lemma:round_matching_violation}
Let $\{\z_h\}_{h\in \Hos}$ be a stable matching respect to $\{\alpha_k \}_{k\in K}$ and $\{\beta_h\}_{h\in \Hos}$. 
Then, there exists an algorithm rounding $\{\z_h\}_{h\in \Hos}$ to $\{\z'_h\}_{h\in \Hos}$ as in Lemma~\ref{lemma:round_matching} such that
\[
|c'_h-c_h|\leq 2|K_h|+1\quad \forall h \in \Hos  
\]
\[
|\alpha'_k - \alpha_k| \le 2|K_h|+2  \quad \forall k \in K_h.
\]
\end{lemma}

The proof follows a standard iterative rounding argument; for completeness, see Appendix~\ref{app:round_matching_violation}.

Finally, we prove Theorem~\ref{thm:mainmatch} as follows.
\begin{itemize}
    \item Construct an MEO with appropriate parameters; existence follows from Theorem~\ref{thm:meo_existence}.
    \item By Theorem~\ref{thm:matching_exists}, the MEO induces a fractional stable matching.
    \item Apply Lemmas~\ref{lemma:round_matching} and~\ref{lemma:round_matching_violation} to obtain an integral stable matching with the desired approximation guarantees.
\end{itemize}


\section{Conclusion}

We study a stable matching problem in which institutions evaluate applicants using multiple rankings. To model institutional decision making, we introduce a social choice notion of \emph{acceptability} at each institution. This approach departs from standard models, which typically rely on exogenous tie-breaking rules or reserve policies. Instead, we internalize the aggregation of rankings into the institution’s own decision process.
By doing so, our framework bridges two strands of the literature: social choice and stable matching. It allows institutional preferences to be derived endogenously from underlying rankings, rather than imposed externally, and thus provides a more faithful representation of collective decision making in many real-world matching markets. 

From a technical perspective, our main methodological contribution is the use of equilibrium relaxation based on ordinal, rather than cardinal, preferences. A distinctive feature of our approach is the encoding of ordinal preferences through random demand, which enables a convexification of the problem while preserving the underlying preference structure. This technique provides a new way to combine ordinal information with equilibrium-based methods.

We believe this approach is broadly applicable beyond the setting studied here. In particular, the use of random demand to convexify ordinal preference problems may be useful in other allocation, matching, and market design contexts. Exploring these applications, as well as strengthening approximation guarantees and extending the model to richer institutional constraints, are promising directions for future work.

An important limitation of our work concerns computation. While our results establish existence and approximation guarantees, computing the corresponding equilibria and stable matchings may be computationally demanding in large markets or under complex institutional constraints. Developing efficient algorithms, identifying tractable special cases, and studying the computational complexity of our solution concepts are natural directions for future research.

\bibliographystyle{plainnat}
\bibliography{sample}

@article{abdulkadiroglu2003school,
  title   = {School Choice: A Mechanism Design Approach},
  author  = {Abdulkadiro{\u{g}}lu, Atila and S{\"o}nmez, Tayfun},
  journal = {American Economic Review},
  volume  = {93},
  number  = {3},
  pages   = {729--747},
  year    = {2003},
  
}

@article{ashlagi2019assigning,
  title={Assigning more students to their top choices: A comparison of tie-breaking rules},
  author={Ashlagi, Itai and Nikzad, Afshin and Romm, Assaf},
  journal={Games and Economic Behavior},
  volume={115},
  pages={167--187},
  year={2019},
  publisher={Elsevier},
  doi={10.1016/j.geb.2019.02.015}
}

@article{nguyen2021stability,
  title={Stability in Matching Markets with Complex Constraints},
  author={Nguyen, Hai and Nguyen, Th{\`a}nh and Teytelboym, Alexander},
  journal={Management Science},
  volume={67},
  number={12},
  pages={7438--7454},
  year={2021},
  publisher={INFORMS},
  doi={10.1287/mnsc.2020.3869}
}

@article{abdulkadiroglu2009strategy,
  title={Strategy-proofness versus efficiency in matching with indifferences: Redesigning the NYC high school match},
  author={Abdulkadiro{\u{g}}lu, Atila and Pathak, Parag A. and Roth, Alvin E.},
  journal={American Economic Review},
  volume={99},
  number={5},
  pages={1954--1978},
  year={2009}
}

@article{erdil2008matter,
  title={What's the matter with tie-breaking? Improving efficiency in school choice},
  author={Erdil, Aytek and Ergin, Haluk},
  journal={American Economic Review},
  volume={98},
  number={3},
  pages={669--689},
  year={2008},
  doi={10.1257/aer.98.3.669}
}

@article{hafalir2013effective,
  title={Effective affirmative action in school choice},
  author={Hafalir, Isa E. and Yenmez, M. Bumin and Yildirim, Muhammed Ali},
  journal={Theoretical Economics},
  volume={8},
  number={2},
  pages={325--363},
  year={2013}
}

@article{ehlers2014school,
  title={School choice with controlled choice constraints: Hard bounds versus soft bounds},
  author={Ehlers, Lars and Hafalir, Isa E. and Yenmez, M. Bumin},
  journal={Journal of Economic Theory},
  volume={153},
  pages={648--683},
  year={2014}
}

@article{sonmez2022affirmative,
  title={Affirmative action in India via vertical, horizontal, and overlapping reservations},
  author={S{\"o}nmez, Tayfun and Yenmez, M. Bumin},
  journal={Econometrica},
  volume={90},
  number={3},
  pages={1143--1176},
  year={2022},
  doi={10.3982/ECTA17788}
}

@incollection{nguyen2023complementarities,
  title       = {Complementarities and Externalities},
  author      = {Nguyen, Thanh and Vohra, Rakesh},
  booktitle   = {Online and Matching-Based Market Design},
  editor      = {Echenique, Federico and Immorlica, Nicole and Vazirani, Vijay V.},
  publisher   = {Cambridge University Press},
  year        = {2023},
  pages       = {323--342},
  isbn        = {9781108831994},
  doi         = {10.1017/9781108937535.015}
}

@inproceedings{chen2018stable,
  author       = {Jiehua Chen and Rolf Niedermeier and Piotr Skowron},
  title        = {Stable Marriage with Multi‑Modal Preferences},
  booktitle    = {Proceedings of the 2018 ACM Conference on Economics and Computation (EC 2018)},
  series       = {ACM EC},
  pages        = {269--286},
  year         = {2018},
  month        = {June}
}

@article{boehmer2025proportional,
  title        = {Proportional representation in matching markets: selecting multiple matchings under dichotomous preferences},
  author       = {Boehmer, Niclas and Brill, Markus and Schmidt-Kraepelin, Ulrike},
  journal      = {Social Choice and Welfare},
  volume       = {64},
  number       = {1-2},
  pages        = {179--220},
  year         = {2023},
  publisher    = {Springer},
  doi          = {10.1007/s00355-023-01453-7},
}

@inproceedings{chen2021fractional,
  title={Fractional Matchings under Preferences: Stability and Optimality},
  author={Chen, Jiehua and Roy, Sanjukta and Sorge, Manuel},
  booktitle={Proceedings of the Thirtieth International Joint Conference on Artificial Intelligence (IJCAI)},
  pages={89--95},
  year={2021}
}

@inproceedings{wen2022position,
  title={Position-Based Matching with Multi-Modal Preferences},
  author={Wen, Yinghui and Zhou, Aizhong and Guo, Jiong},
  booktitle={Proceedings of the International Conference on Autonomous Agents and Multi-Agent Systems (AAMAS)},
  year={2022}
}

@article{masatlioglu2005reference,
  title={Reference-Dependent Preferences},
  author={Masatlioglu, Yusufcan and Ok, Efe A.},
  journal={The Quarterly Journal of Economics},
  volume={120},
  number={4},
  pages={1415--1448},
  year={2005},
  publisher={Oxford University Press}
}

@article{bordalo2013salience,
  title={Salience Theory of Choice under Risk},
  author={Bordalo, Pedro and Gennaioli, Nicola and Shleifer, Andrei},
  journal={The Journal of Political Economy},
  volume={121},
  number={5},
  pages={803--843},
  year={2013},
  publisher={University of Chicago Press}
}

@article{koszegi2006model,
  title={A Model of Reference-Dependent Preferences},
  author={K{\"o}szegi, Botond and Rabin, Matthew},
  journal={The Quarterly Journal of Economics},
  volume={121},
  number={4},
  pages={1133--1165},
  year={2006},
  publisher={Oxford University Press}
}

@article{kiraly2012degree,
  title={Degree bounded matroids and submodular flows},
  author={Kir{\'a}ly, Tam{\'a}s and Lau, Lap Chi and Singh, Mohit},
  journal={Combinatorica},
  volume={32},
  number={6},
  pages={703--720},
  year={2012},
  publisher={Springer}
}

@article{matching2018,
  title={Near-Feasible Stable Matching with Couples},
  author={Nguyen, Thanh and Vohra, Rakesh},
  journal={American Economic Review},
  volume={106},
  number={11},
  pages={3154-3169},
  year={2018},
}

@article{allocation2021,
  title={Allocation with Weak Priorities and General Constraints},
  author={Lin, Young-San and Nguyen, Hai and Nguyen, Thanh and Altinkemer,Kemal},
  journal={Operation Research },
  volume={70},
  number={5},
  pages={2597-3033},
  year={2022},

}

@article{gale_shapley_1962,
  title        = {College Admissions and the Stability of Marriage},
  author       = {Gale, David and Shapley, Lloyd S.},
  journal      = {The American Mathematical Monthly},
  volume       = {69},
  number       = {1},
  pages        = {9--15},
  year         = {1962},
  doi          = {10.1080/00029890.1962.11989827},
  publisher    = {Mathematical Association of America},
 
}

@article{hatfield_milgrom_2005,
  title        = {Matching with Contracts},
  author       = {Hatfield, John William and Milgrom, Paul R.},
  journal      = {American Economic Review},
  volume       = {95},
  number       = {4},
  pages        = {913--935},
  year         = {2005},
  doi          = {10.1257/0002828054825466},
  
}

@article{school_japan,
  author = {Kitahara, Minoru and Okumura, Yasunori},
  title  = {School Choice with Multiple Priorities},
  year   = {2024},
  journal = {arXiv preprint arXiv:2308.04780}, 
}

@inproceedings{voting_us_26,
  title={A Few Good Choices},
  author={Song, Haoyu and Nguyen, Thanh and Lin, Young-San},
  booktitle={Proceedings of the 2026 Annual ACM-SIAM Symposium on Discrete Algorithms (SODA)},
  pages={4861-4874},
  year={2025}
}

@article{nguyen-song,
title={Approximate core of Participatory Budgeting},
  author={Haoyu Song and Thanh Nguyen},
  journal={Manuscript},
  year={2024}
}

@article{nguyen2016assignment,
  title={Assignment problems with complementarities},
  author={Nguyen, Thanh and Peivandi, Ahmad and Vohra, Rakesh},
  journal={Journal of Economic Theory},
  volume={165},
  pages={209--241},
  year={2016},
  publisher={Elsevier}
}

@article{path_independence_81, 
title = {General theory of best variants choice: Some aspects}, 
author = {Aizerman, M. and Malishevski, A.}, 
journal = {IEEE Transactions on Authomatic Control}, 
volume = {26}, 
issue = {5}, 
pages = {1030-1040}, 
year = {1981},
publisher = {IEEE}
}

@book{book11iterative,
  title={Iterative Methods in Combinatorial Optimization},
  author={Lau, Lap Chi and Ravi, R and Singh, Mohit},
  year={2011},
  publisher={Cambridge University Press}
}

@article{choice_and_matching, 
title = {Choice and Matching}, 
author = {Chamber, Christoper P. and Yenmez, M.Bumin}, 
journal = {American Economic Journal:Microeconomics}, 
volume = {9}, 
issue = {3}, 
pages = {126-147}, 
year = {2017}
}

@article{chamberlin1983representative,
  title={Representative Deliberations and Representative Decisions: Proportional Representation and the Borda Rule},
  author={Chamberlin, J. R. and Courant, P. N.},
  journal={American Political Science Review},
  volume={77},
  number={3},
  pages={718--733},
  year={1983}
}

@article{monroe1995fully,
 author = {Burt L. Monroe},
 journal = {The American Political Science Review},
 number = {4},
 pages = {925--940},
 publisher = {[American Political Science Association, Cambridge University Press]},
 title = {Fully Proportional Representation},
 urldate = {2026-03-04},
 volume = {89},
 year = {1995}
}

@article{elkind2015condorcet,
  title={Condorcet winning sets},
  author={Elkind, E. and Lang, J. and Saffidine, A.},
  journal={Social Choice and Welfare},
  volume={45},
  pages={1--27},
  year={2015}
}

@inproceedings{aziz2017justified,
  title={Justified representation in approval-based committee voting},
  author={Aziz, H. and Brill, M. and Conitzer, V. and Elkind, E. and Freeman, R. and Walsh, T.},
  booktitle={Proceedings of the 31st AAAI Conference on Artificial Intelligence (AAAI)},
  pages={784--790},
  year={2017}
}

@inproceedings{CharikarLassotaRamakrishnanVettaWang2025,
  author       = {Moses Charikar and Alexandra Lassota and Prasanna Ramakrishnan and Adrian Vetta and Kangning Wang},
  title        = {Six Candidates Suffice to Win a Voter Majority},
  booktitle    = {Proceedings of the 57th Annual ACM Symposium on Theory of Computing (STOC)},
  year         = {2025},
  pages        = {1590--1601},
  publisher    = {Association for Computing Machinery},
  doi          = {10.1145/3717823.3718235},
  
}

\newpage 
\begin{center}
    APPENDIX
\end{center}
\appendix

\section{Rounding Lemmas}\label{app:rounding}
\subsection{Proof of Lemma~\ref{lemma:round}}
First, we show the following.
\begin{claim}\label{cor:support_inside}
For any student $i \in C^*$, her degree of support with respect to $C^*$ and $\vec{\alpha}'$ is at least her degree of \emph{weak} support with respect to $\z$ and $\vec{\alpha}$, that is,
\[
\mathrm{Support}_h^{\vec{\alpha}'}(i, C^*) \;\ge\; \mathrm{w\text{-}Support}_h^{\vec{\alpha}}(i, \z).
\]
\end{claim}

\begin{proof}    

To prove this, we show that for any student $i \in C^*$, if a committee member $k$ $\alpha_k$-weakly approves $i$ with respect to $\z$, then $k$ also $\alpha'_k$-approves $i$ with respect to $C^*$.

 For any $i'\in C^*$ with $i'\succeq_k i$, we show that $i'\in U_k(\z,\vec{\alpha})$. To start with, having $i'\in C^*$ implies $z_{i'}>0$. Then, given that $i'\succeq_k i$ and $i\succeq_k i_{k,\z}^{\vec{\alpha}}$, it follows from transitivity that $i'\succeq_k i_{k,\z}^{\vec{\alpha}}$ and therefore we have $i'\in U_k(\z,\vec{\alpha})$. 

 Now, we can bound the total number of students whom $k$ weakly-prefers to $i$ as:
 $$ |i'\in C^*:i'\succeq_k i| = \sum_{i'\in C^*:i'\succeq_k i}z'_{i'} \leq \sum_{i'\in U_k(\z,\vec{\alpha})}z'_{i'} = \alpha'_k,$$
 This implies that $k$ $\alpha_k'$-approves $i$ with respect to $C^*$ as desired. 
\end{proof}

Next, we consider the set of students who are left out of the final allocation and derive results about their level of support. 
\begin{claim}\label{cor:support_outside}
For any student $i\notin C^*$, it holds that her degree of support with respect to $C^*$ and $\vec{\alpha}'$ is at most her degree of \textit{strong} support with respect to $\z$ and $\vec{\alpha}$ i.e.
$$\mathrm{Support}_h^{\vec{\alpha}'}(i,C^*)\leq \mathrm{s\text{-}Support}_h^{\vec{\alpha}}(i,\z).$$ 
\end{claim}

\begin{proof} 
To prove this, we show that
for any student $i\notin C^*$, if a committee member $k$ $\alpha'_k$-approves $i$ with respect to $C^*$, then $k$ must also strongly $\alpha_k$-approves $i$ with respect to $\z$.

Given that $k$ $\alpha'_k$-approves $i$ with respect to $C^*$, the total number of students in $C^*$ ranked \textit{strictly above} $i$ is at most $\alpha'_k - 1$. Then, let $i^*$ be voter $k$'s lowest ranked student in $C^*\cap U_k(\z,\vec{\alpha})$. It necessarily follows that $i\nprec_k i^*$, because having $i\prec_k i^*$ would imply that the entire set $U_k(\z,\vec{\alpha})\cap C^*$ is \textit{strictly above} $i$ and the total number of students strictly above $i$ is at least $\alpha_k'$, yielding a contradiction. Furthermore, since $i\notin C^*$, we have $i\neq i^*$. So, it must hold true that $i\succ_k i^*.$

Now, we know $i^*\succeq_k i^{\vec{\alpha}}_{k,\z}$ because $i^{\vec{\alpha}}_{k,\z}$ is by definition $k$'s lowest ranked student in $U_k(\z,\vec{\alpha})$. It follows from transitivity that $i\succ_k i_{k,\z}^{\vec{\alpha}}$ and therefore $k$ strongly $\alpha_k$-approves $i$ with respect to $\z$. 

\end{proof}

Thus, Claims~\ref{cor:support_inside} and~\ref{cor:support_outside}, together with the fact that $\z$ is an acceptable fractional set, imply the result.

\subsection{Proof of Lemma~\ref{lemma:iterround}}

\begin{proof}
We round $\z$ to an integral vector $\z'$ via iterative rounding. 
Fix all variables $z_i \in \{0,1\}$ by setting $z'_i := z_i$, and restrict attention to the remaining variables; hence, without loss of generality, $0<z_i<1$ for all $i\in S$.

Let
$a_k := \lfloor \sum_{i \in U^+_{k}(\z,\vec{\alpha})} z_i \rfloor$ and
$b_k := \lceil \sum_{i \in U^+_{k}(\z,\vec{\alpha})} z_i \rceil$.
Consider the following linear program:
\begin{align*}
\sum_{i \in S} x_i &= c_h && &\text{(capacity),} \\
a_k \;\le\; \sum_{i \in U^+_{k}(\z,\vec{\alpha})} x_i 
&\le\; b_k && \forall k \in K_h & \text{(committee demand),} \\
0 \;\le\; x_i &\le\; 1 && \forall i \in S & \text{(bounds).}
\end{align*}

Notice that, since $z$ is acceptable, for each $i \in S$ with $z_i < 1$, we have
$$\mathrm{s\text{-}Support}_h^{\vec{\alpha}}(i,\x) = \bigl|\{k\in K_h \mid i\in U_k^{+}(\x, \vec{\alpha})\}\bigr| \le \beta.$$
That is, $i$ is in at most $\beta$ number of sets $U^+_{k}(\z,\vec{\alpha})$, we can apply 
\cite{kiraly2012degree}: there exists an $0-1$ vector $\z'$ with
$\sum_{i\in S} z'_i=c_h$ and
\[
a_k-2\beta+1 \le \sum_{i \in U^+_{k}(\z,\vec{\alpha})} z'_i \le b_k+2\beta-1
\quad \forall k\in K_h .
\]

Using \eqref{eq:D2}, $a_k=\alpha_k-1$ and $b_k=\alpha_k$, hence
\[
\alpha_k-2\beta \le \sum_{i \in U^+_{k}(\z,\vec{\alpha})} z'_i
\le \alpha_k+2\beta-1 .
\]

Finally, since
$\sum_{i \in U_k(\z,\vec{\alpha})} z'_i -1
\le \sum_{i \in U^+_{k}(\z,\vec{\alpha})} z'_i
\le \sum_{i \in U_k(\z,\vec{\alpha})} z'_i$,
we obtain
\[
\alpha_k-2\beta \le \sum_{i \in U_k(\z,\vec{\alpha})} z'_i
\le \alpha_k+2\beta \qquad \forall k\in K_h .
\]
This shows that 
\[
\alpha_k-2\beta \le \alpha_k'
\le \alpha_k+2\beta \qquad \forall k\in K_h .
\]
\end{proof}

\section{Equilibrium Existence}
\subsection{Proof of Theorem~\ref{thm:LEO_existence}}
\begin{proof} [Proof of Theorem~\ref{thm:LEO_existence}] We will give the fixed point and a sketch of the proof here. The proof is virtually the same as in \citep{nguyen-song,voting_us_26}.   

Let us denote $$\Omega := \prod_{k\in K_h} [0,1]^{n+1}\times [0,1]^n\times  \prod_{k\in K_h} [0,1]^{n+1}.$$

We consider the following set-valued mapping $\mathcal{F}:\Omega \rightrightarrows 2^{\Omega}$ with
$\mathcal{F}(\y,\z,\prices) = (\y',\z',\prices') $ s.t.

\begin{subequations}
\begin{align}
y'_{k,i}&= \Pr[\mathcal{D}_{k}(\prices_{k}, \randbud{}) =i] & \forall i \in S\cup\{\emptyset\}, k\in K_h\label{eq:fp-y}\\
z_h' &= \mathcal{D}_h(\{\prices_k\}_{k\in K_h}) \label{eq:fp-z} \\ 
p'_{k,i} &= 
\begin{cases}
\max\left\{\min\left\{1,p_{k,i}+(\alpha_{k}\cdot y_{k,i}-z_{i})\right\} ,0 \right\} \text{  if } i \neq \emptyset \\
0 \text{  if } i = \emptyset
\end{cases} & \forall i \in S\cup\{\emptyset\}, k\in K_h \label{eq:fp-pk}
\end{align}
\end{subequations}

We can show that $\mathcal{F}$ satisfies the pre-requisites of Kakutani's fixed point theorem and therefore a fixed point $(\y,\z,\prices)$ exists. The detailed argument can be found in \citep{nguyen-song,voting_us_26}. 

We now prove that the fixed point meets the conditions of a matching equilibrium. Conditions (\ref{LEO-condition1}) and (\ref{LEO-condition2}) are satisfied by definition. We are left to prove that the fixed point meets condition (\ref{LEO-condition3}). 

To show $\alpha_k\cdot y_{k,i}\leq z_{i}$ as in the first part of condition (\ref{LEO-condition3}), we observe that if $\alpha_k\cdot y_{k,i}>z_{i}$ at a fixed point, it necessarily follows that $p_{k,i} = 1$ because otherwise we would have $p'_{k,i}>p_{k,i}$. For any continuous $\randbud{}$ supported on $[0,1]$, it holds that $\Pr_{b\sim \randbud{}}(b = 1) = 0$ and therefore $y_{k,i} = \Pr(\mathcal{D}_k({\prices_k,\randbud{}})= i)=0$, contradicting to the assumption of $\alpha_k\cdot y_{k,i}>z_{i}$. For the second part, if $\alpha_k\cdot y_{k,i}<z_{i}$, it necessarily follows that $p_{k,i} = 0$ because otherwise we would have $p'_{k,i}<p_{k,i}$. This completes the proof. 
\end{proof}

\subsection{MEO existence} 

We now show MEO exists when the cumulative function of the budget distribution is continuous. 
\begin{proof}[Proof of Theorem \ref{thm:meo_existence}]
It suffices to give the corresponding mapping. The existence of the fixed point follows in the same way as \citep{nguyen-song,voting_us_26}. The proof that the fixed point satisfies the desirable properties work in the same way as in the outline given in Theorem \ref{thm:LEO_existence}. 

Let us denote $$\Omega := \prod_{i=1}^n [0,1]^{m+1}\times \prod_{k\in K} [0,1]^{n+1}\times \prod_{i=1}^m [0,1]^n\times \prod_{i=1}^n [0,1]^{m+1}\times\prod_{k\in K} [\delta,1]^{n+1}.$$

We consider the following set-valued mapping $\mathcal{F}:\Omega \rightrightarrows 2^{\Omega}$ with
$\mathcal{F}(\x,\y,\z,\bold{q},\prices) = (\x',\y',\z',\bold{q}',\prices') $ s.t.

\begin{subequations}
\begin{align}
x'_{i,h} &= \Pr[\mathcal{D}_{i}(\prices_{i}, \randbud{}) =h] & \forall i \in S, h\in \Hos\cup\{\emptyset\}, \label{eq:fp-x} \\
y'_{k,i}&= \Pr[\mathcal{D}_{k}(\prices_{k}, \randbud{}) =i] & \forall i \in S\cup\{\emptyset\}, k\in K\\
z_h' &= \mathcal{D}_h(\{\prices_k\}_{k\in K_h}, \{\bold{q}_i\}_{i\in D})&\forall h\in H \label{eq:fp-z} \\ 
q'_{i,h} &= 
\begin{cases}
\max\left\{\min\left\{1,q_{i,h}+(x_{i,h}-z_{h,i})\right\} ,0 \right\} \text{  if } h \neq \emptyset \\
0 \text{  if } h = \emptyset
\end{cases} & \forall i \in S, h \in \Hos\cup \{\emptyset\}\label{eq:fp-pi}\\
p'_{k,i} &= 
\begin{cases}
\max\left\{\min\left\{1,p_{k,i}+(\alpha_{k}\cdot y_{k,i}-z_{h,i})\right\} ,\delta \right\} \text{  if } i \neq \emptyset \\
\delta \text{  if } i = \emptyset
\end{cases} & \forall i \in S\cup\{\emptyset\}, h \in \Hos, k\in K_h \label{eq:fp-pk}
\end{align}
\end{subequations}

\end{proof}

\section{ Proof of Theorem~\ref{thm:exist}}\label{app:LEO-accept}

For notation brevity, we use $\mathbf{U}[a, b]$ to denote the uniform distribution on $[a,b]$.

\begin{claim}\label{clm:prices_LEO}
Let $h$ be a school with capacity $c_h$, a set of committee members $K_h$, and a ranking vector $\vec{\alpha} \in [0,c_h]^{|K_h|}$. Given $\varepsilon \in (0,1)$, let $(\y, \z, \prices)$ be the corresponding LEO with $\randbud{} = \mathbf{U}[1-\varepsilon, 1]$, and define
\[
p^* = \min\Big\{\sum_{k \in K_h} p_{k,i} : i \in S, z_i > 0\Big\}.
\]
Then the following hold:
\begin{enumerate}
    \item $p^* \le \frac{\sum_{k \in K_h} \alpha_k}{c_h}$\label{claim-price-condition1};
    \item for each $i$ with $z_i < 1$, $\sum_{k \in K_h} p_{k,i} \le p^*$\label{claim-price-condition2}.
\end{enumerate}
\end{claim}

\begin{proof}
We first prove condition (\ref{claim-price-condition1}) by contradiction. Suppose that $p^* > \frac{\sum_{k \in K_h} \alpha_k}{c_h}$. Let $R(\x)$ denote the total revenue collected from a fractional allocation $\x$.  

By condition (\ref{LEO-condition3}) of LEO, the revenue collected from committee member $k$ is at most $\alpha_k$ times the expected budget. That is,
\begin{equation}\label{eq:individual_revenue_bound}
\sum_{i \in S} p_{k,i} z_i = \sum_{i \in S} p_{k,i} (\alpha_k \cdot y_{k,i}) = \alpha_k \sum_{i \in S} p_{k,i} y_{k,i} \le \alpha_k \cdot \mathbb{E}_{b \sim \randbud{}}[b] \le \alpha_k.
\end{equation}

Summing over all committee members, the total revenue is bounded above:
\begin{equation}\label{eq:revenue_bound_upper}
R(\z) = \sum_{k \in K_h} \sum_{i \in S} p_{k,i} z_i \le \sum_{k \in K_h} \alpha_k.
\end{equation}

On the other hand, by definition of $p^*$, we have
\begin{equation}\label{eq:revenue_bound_lower}
R(\z) = \sum_{i \in S} \Big(\sum_{k \in K_h} p_{k,i}\Big) z_i \ge \sum_{i \in S, z_i>0} p^* z_i \ge p^* \cdot c_h > \frac{\sum_{k \in K_h} \alpha_k}{c_h} \cdot c_h = \sum_{k \in K_h} \alpha_k,
\end{equation}
which contradicts (\ref{eq:revenue_bound_upper}). Hence $p^* \le \frac{\sum_{k \in K_h} \alpha_k}{c_h}$.

For condition (\ref{claim-price-condition2}), suppose by contradiction that there exists a student $\hat{i}$ with $z_{\hat{i}} < 1$ and $\sum_{k \in K_h} p_{k,\hat{i}} > p^*$. Let $i^*$ be a student with $z_{i^*} > 0$ achieving the minimal cumulative price $p^*$. Define 
\[
\delta = \min(1 - z_{\hat{i}}, z_{i^*}),
\]
and construct a new allocation $\hat{\z}$ by
\[
\hat{z}_i =
\begin{cases}
z_i & i \notin \{i^*, \hat{i}\} \\
z_{i} + \delta & i = \hat{i} \\
z_{i} - \delta & i = i^*.
\end{cases}
\]

The change in revenue is
\[
R(\hat{\z}) - R(\z) = \left(\sum_{k \in K_h} p_{k,\hat{i}} - p^*\right) \delta > 0,
\]
which contradicts the revenue-maximizing property of $\z$. Therefore, $\sum_{k \in K_h} p_{k,i} \le p^*$ for all $i$ with $z_i < 1$, completing the proof.
\end{proof}

Our next step is to derive the relationship between cumulative price and the degree of fractional support at a LEO. For this purpose, we need the following technical lemma showing that each committee's $\alpha_k$-rank student is also her favorite student with a price of at most $1-\varepsilon$. The proof relies on the relationship between individual random demand and the central allocation at a LEO as seen in Condition (\ref{LEO-condition3}) of Definition \ref{def:LEO}. 

\begin{claim}\label{lemma:boundary_student}
Let $h$ be a school with a capacity $c_h$, a set of committee members $K_h$, and a ranking vector $\vec{\alpha}\in [1,c_h]^{|K_h|}$. Given $\varepsilon\in(0,1)$, let $(\y,\z,\prices)$ be the corresponding LEO with $\randbud{}=\mathbf{U}[1-\varepsilon,1]$. Then, for each committee member $k\in K_h$, her $\alpha_k$-rank student with respect to $\z$ is also her favorite student with a price of at most $1-\varepsilon$  i.e. 
$$i_{k,\z}^{\vec{\alpha}} = \max_{\succ_k}\{i\in S:p_{k,i}\leq 1-\varepsilon\}.$$
\end{claim}
\begin{proof}
To ease the burden of notation, we will write $i^*:= \max_{\succ_k}\{i\in S:p_{k,i}\leq 1-\varepsilon\}$. Suppose by contradiction that $i_{k,\z}^{\vec{\alpha}}\neq i^*$. Since we assume that each committee member has a strict order, we break into two cases and show that each case leads to a contradiction.  
\begin{itemize}
\item [\textbf{Case 1:}]$i_{k,\z}^{\vec{\alpha}}\prec_k i^*$. In this case, we have $\sum_{i\succeq i^*}y_{k,i} = 1$. Then, by condition (\ref{LEO-condition3}) of the definition of LEO, we can derive  
$$ \sum_{i\succeq i^*}z_i \geq \alpha_k\cdot \sum_{i\succeq i^*}y_{k,i} = \alpha_k.$$
As a result, $i_{k,\z}^{\vec{\alpha}}$ would not be $k$'s favorite student with total weight above being at least $\alpha_k$, contradicting to the definition of $\alpha_k$-rank student. 
\end{itemize}

\item [\textbf{Case 2:}]$i_{k,\z}^{\vec{\alpha}}\succ_k i^*$.  
By assumption, $p_{k,i}>1-\varepsilon$ for any $i\succ_k i^*$. In particular, it implies that there does not exist any $i'\succ_k i^*$ with $\alpha_k\cdot y_{k,i}<z_{i}$: if such an $i'$ exists, it necessarily follows that $p_{k,i'} = 0\leq 1-\varepsilon$, yielding a contradiction. Therefore, we can obtain $\alpha_k\cdot y_{k,i} = z_{i}$ for any $i\succ_k i^*$. 

Now, we can see that $y_{k,i^*}>0$ because, given the discussion above, $i^*$ is $k$'s favorite affordable student under the budget $b\in \big[1-\varepsilon, \min\{p_{k,i}:i\succ_k i^*\}\big)$. 
It then follows that the total weight above $i_{k,\z} ^{\vec{\alpha}}$ in $k$'s random demand is \textit{strictly} smaller than 1 i.e.: 
$$ \sum_{i\succeq i_{k,\z}^{\vec{\alpha}}}y_{k,i}\leq \sum_{i\neq i^*}y_{k,i} = 1 - y_{k,i^*}<1.$$
We can then derive a \textit{strict} upper bound on the total weight above $i_{k,\vec{z}}^{\vec{\alpha}}$ in $\z$:  
$$ \sum_{i\succeq i_{k,\z}^{\vec{\alpha}}}z_{i} =\alpha_k\cdot \sum_{i\succeq i_{k,\z}^{\vec{\alpha}}}y_{k,i}<\alpha_k.$$
Now, because $|\z|_1 = c_h\geq \alpha_k$, it follows $i_{k,\z}\neq \emptyset$. Then, it holds by definition that $\sum_{i\succeq i_{k,\z}^{\vec{\alpha}}}z_{i}$ is \textit{at least} $\alpha_k$, which contradicts to the derivation above and thus completes the proof. 
\end{proof}
With Claim \ref{lemma:boundary_student}, we are able to translate the results on cumulative price to results on the level of support. In particular, the following claim shows that a LEO with $\randbud{} = \mathbf{U}[1-\varepsilon, 1 ]$ can be seen as an ``approximately" acceptable fractional set, where the condition of no blocking is violated by a factor of at most $\frac{1}{1-\varepsilon}$.

\begin{claim}\label{clm:LEO_approx}
Let $h$ be a school with a capacity $c_h$, a set of committee members $K_h$, and a ranking vector $\vec{\alpha}\in [0,c_h]^{|K_h|}$. Given $\varepsilon\in(0,1)$, let $(\y,\z,\prices)$ be the corresponding LEO with $\randbud{}=\mathbf{U}[1-\varepsilon,1]$. Then, there exists a $\beta\leq \frac{\sum_{k\in K_h}\alpha_k}{c_h}$ s.t. the following statements are true for $\z$: 
\begin{enumerate}
\item $|\z|_1 = c_h$; \label{LEO-approx-condition1}
\item for each $i$ with $z_i>0$, $\mathrm{w\text{-}Support}_h^{\vec{\alpha}}(i,\z)\geq \beta$; \label{LEO-approx-condition2} 
\item for each $i$ with $z_i<1$, $\mathrm{s\text{-}Support}_h^{\vec{\alpha}}(i,\z)\leq \lfloor \frac{\beta}{1-\varepsilon}\rfloor.$\label{LEO-approx-condition3}
\end{enumerate}

\end{claim}
\begin{proof}
Taking $p^* = \min\{\sum_{k\in K_h}p_{k,i}:i\in S,z_i>0\}$ from Claim \ref{clm:prices_LEO}, we set $\beta = p^*\leq \frac{\sum_{k\in K_h}\alpha_k}{c_h}$. We can see that condition (\ref{LEO-approx-condition1}) is trivially satisfied by the definition of LEO. We are now left to show conditions (\ref{LEO-approx-condition2}) and (\ref{LEO-approx-condition3}). 

Now, we show (\ref{LEO-approx-condition2}). Pick an arbitrary student $i'$ with $z_{i'}>0$. Then, we start by observing that $p_{k,i'}>0$ for a committee member $k$ only if $k$ weakly supports $i'$. To see this, suppose a committee member $k'$ does not weakly support $i'$ and therefore $i'\prec_{k'} i_{k',\z}^{\vec{\alpha}}$. Because the price of $i_{k,\z}^{\vec{\alpha}}$ for $k$ is at most $1-\varepsilon$ by Claim \ref{lemma:boundary_student}, the random demand of $k'$ will never pick $i'$ and therefore $y_{k',i'} = 0$. It then follows that $\alpha_{k'}y_{k',i'}<z_{i'}$ and therefore $p_{k',i'} =0$. 

Based on the observation above, Claim \ref{clm:prices_LEO}, and the fact that personalized prices are at most 1, we can derive 
$$|k:i'\succeq_k i_{k,\z}^{\vec{\alpha}}|\geq \sum_{k:i'\succeq_k i_{k,\z}^{\vec{\alpha}}}p_{k,i'}\geq p^*,$$
from which we obtain $\mathrm{w\text{-}Support}_h^{\vec{\alpha}}(i',\z)\geq p^*=\beta$ as desired. 

Next, we show $(\ref{LEO-approx-condition3})$. Pick an arbitrary student $\hat{i}$ with $z_{\hat{i}}<1$. We observe that for any committee member $k$ who strongly supports $\hat{i}$, it holds $p_{k,\hat{i}}>1-\varepsilon$ because, by Claim \ref{lemma:boundary_student}, $i_{k,\z}^{\vec{\alpha}}$ is $k$'s favorite student with a price of at most $1-\varepsilon$. Given that the cumulative price for $\hat{i}$ is at most $p^*$ by Claim \ref{clm:prices_LEO}, we can derive 
$$(1-\varepsilon)\cdot |k:\hat{i}\succ_k i_{k,\z}^{\vec{\alpha}}|\leq \sum_{k:\hat{i}\succ_k i_{k,\z}^{\vec{\alpha}}}p_{k,\hat{i}}\leq \sum_{k\in K_h}p_{k,\hat{i}}\leq p^*,$$
from which we obtain and $|k:\hat{i}\succ_k i_{k,\z}^{\vec{\alpha}}|= \mathrm{s\text{-}Support}_h^{\vec{\alpha}}(i,\z)\leq \lfloor \frac{p^*}{1-\varepsilon}\rfloor$ as desired. 
\end{proof}

Finally, Theorem~\ref{thm:exist} follows from Claim~\ref{clm:LEO_approx} by choosing $\varepsilon>0$ sufficiently small. Specifically, we can pick $\varepsilon$ so that, for all 
\(\beta \le \frac{\sum_{k\in K_h} \alpha_k}{c_h}\),  
\[
\frac{\beta}{1-\varepsilon} - \beta < 1 \quad 
\]  
It then follows that  
\[
 \left\lfloor \frac{\beta}{1-\varepsilon} \right\rfloor \le \lceil \beta \rceil \le  \left\lceil \frac{\sum_{k\in K_h} \alpha_k}{c_h} \right\rceil.
\]  

Let $\beta$ be the parameter in Claim~\ref{clm:LEO_approx}, and let $(\y,\z,\prices)$ be the corresponding LEO. Define  
\[
\beta^* := \left\lfloor \frac{\beta}{1-\varepsilon} \right\rfloor.
\]  
We then have that $\z$ is an acceptable fractional set with respect to $\beta^* \le \left\lceil \frac{\sum_{k\in K_h} \alpha_k}{c_h} \right\rceil$:

\begin{itemize}
    \item Condition~(\ref{LEO-approx-condition1}) of Claim~\ref{clm:LEO_approx} guarantees non-wastefulness.
    \item Condition~(\ref{LEO-approx-condition2}) implies 
    \[
    \mathrm{w\text{-}Support}_h^{\vec{\alpha}}(i,\z) \ge \beta \quad \Rightarrow \quad 
    \mathrm{w\text{-}Support}_h^{\vec{\alpha}}(i,\z) \ge \lceil \beta \rceil \ge \beta^*.
    \]
    \item Condition~(\ref{LEO-approx-condition3}) ensures 
    \[
    \mathrm{w\text{-}Support}_h^{\vec{\alpha}}(i,\z) \le \beta^*.
    \]
\end{itemize}

Hence, all requirements for an acceptable fractional set are satisfied by $\z$.


\section{Proof of Theorem~\ref{thm:matching_exists}}\label{app:thm:matching_exists}

We first give two technical lemmas. The first one shows the total revenue of a school is upper bonded by the total budget of its committee members plus an amount resulting from the lower bound $\delta$ of each committee member's pricing. 

For notation brevity, we use $\mathbf{U}[a, b]$ to denote the uniform distribution on $[a,b]$.

\begin{lemma}\label{lemma:revenue_with_delta}
Let $(\Hos, S, \{c_h\}_{h\in \Hos},\Vec{\alpha}:=\{\alpha_{k}\}_{k\in K})$ be an instance of the matching problem. Let $(\x,\y,\z,\prices, \bold{q})$ be a MEO under $\delta>0$ and $\randbud{}= \mathbf{U}[1-\varepsilon, 1]$ with $\varepsilon>0$. Then, given a school $h\in \Hos$, its total utility $U(\z_h):=\sum_{i \in S}u_{h,i} z_{h,i}$ is upper bounded by the total revenue
$$R(\z_h):= \sum_{k\in K_h}\sum_{i\in S}p_{k,i}\cdot z_{h,i} \leq \big(\sum_{k\in K_h}\alpha_k\big)+ \delta\cdot c_h\cdot |K_h|.$$
\end{lemma}
\begin{proof}

Fix a voter $k\in K_h$, we decompose the set of students as $S_{eq}:=\{i\in S:\alpha_k\cdot y_{k,i} = z_{h,i}\}$ and $S_< := \{i\in S:\alpha_k\cdot y_{k,i}<z_{h,i}\}$. Based on the decomposition and the fact that a $k$'s spending can never exceed her expected budget, we can upper bound the total revenue collected from $k$ as: 
\begin{align*}
\sum_{i\in S}p_{k,i}\cdot z_{h,i} &=  \sum_{i\in S_{eq}}p_{k,i}\cdot z_{h,i} + \sum_{i\in S_{<}}p_{k,i}\cdot z_{h,i} = \alpha_k\cdot \big(\sum_{i\in S_{eq}}p_{k,i}\cdot y_{k,i}\big) + \delta\cdot \sum_{i\in S_<} z_{h,i}\\
&\leq \alpha_k\cdot \mathbb{E}_{b \sim \randbud{}}[b]+\delta\cdot c_h =\alpha_k + \delta\cdot c_h.
\end{align*}

Then, we can bound the total revenue by summing the bounds over the voters in $K_h$: 
$$ R(\z_h) = \sum_{k\in K_h}\big( \sum_{i\in S}p_{k,i}\cdot z_{h,i}\big)\leq \sum_{k\in K_h}(\alpha_k + \delta\cdot c_h)=\big(\sum_{k\in K_h}\alpha_k\big)+ \delta\cdot c_h\cdot |K_h|.$$

Furthermore, we know that the price $q_{i,h}$ of any school $h$ for any student $i$ is at most 1. Therefore, the utility of any student for $h$ is at most the cumulative price of $i$. As a result, we have $U(\z_h)\leq R(\z_h)$ and the bound follows. 
\end{proof}

Similar to the case of $\alpha_k$-rank student, we also give an equivalent characterization of the boundary school based on its price. 

\begin{lemma}\label{lemma:boundary_school_price}
Let $(\Hos, S, \{c_h\}_{h\in \Hos},\Vec{\alpha}:=\{\alpha_{k}\}_{h\in \Hos, k\in K_h})$ be an instance of the matching problem with $\sum_{h\in \Hos}c_h\leq |S|$. Let $(\x,\y,\z,\prices, \bold{q})$ be a MEO under $\delta>0$ and $\randbud{}=\mathbf{U}[1-\varepsilon, 1]$ with $\varepsilon>0$. Then, given a student $s\in S$, her boundary school with respect to $\z$ is also her most favorite school with a price of at most $1-\varepsilon$ i.e.: 
$$h_{i,\z} = \max_{\succ_i}\{h\in \Hos\cup \{\emptyset\}:q_{i,h}\leq 1-\varepsilon\}.$$
\end{lemma}
\begin{proof}

Now we can assume $\sum_{h\in \Hos}z_{h,i} = 1$. Let us use $h^*$ to denote $i$'s most preferred school with a price of at most $1-\varepsilon$. 
\begin{itemize}
\item [\textbf{Case 1:}] $h^* = \emptyset$. This implies $\sum_{h\in \Hos}x_{i,h}<1$ and therefore $\sum_{h\in \Hos}z_{h,i}<1$ by Lemma \ref{lemma:equal_x_z}. Then, by the definition of boundary hospital, we have $h_{i,\z} = \emptyset = h^*$ as desired. 
\item [\textbf{Case 2:}] $h^*\neq \emptyset$. In this case, we will have $x_{i,h^*}>0$ and therefore $z_{h,i^*}>0$ by Lemma \ref{lemma:equal_x_z}. Then, $h^*\succeq_i h_{i,\z}$ by the definition of boundary school. And $h^*\succ h_{i,\z}$ cannot be true because that would imply $x_{i,h_{i,\z}} = z_{h_{i,\z},i} = 0$. As a result, $h_{i,\z}=h^*$ must hold as desired. 
\end{itemize}
\end{proof}

With Lemma \ref{lemma:revenue_with_delta}, we derive the following bounds on utility levels, which can be seen as a counterpart to Claim \ref{clm:prices_LEO}. 

\begin{claim}\label{claim:bound_utility}
Let $(\Hos, S, \{c_h\}_{h\in \Hos},\Vec{\alpha}:=\{\alpha_{k}\}_{h\in \Hos, k\in K_h})$ be an instance of the matching problem with $\sum_{h\in \Hos}c_h\leq |S|$. Let $(\x,\y,\z,\prices, \bold{q})$ be a MEO under $\delta>0$ and $\randbud{}=\mathbf{U}[1-\varepsilon, 1]$ with $\varepsilon>0$. Fix a school $h$, let $u^*_h$ be the smallest utility achieved by any student on the support of $\z_h$ i.e. $u^*_h:=\min_{i\in S}\{u_{h,i}:z_{h,i}>0\}$. Then, the following statements are true: 
\begin{enumerate}
\item $u^*_h\leq \frac{\big(\sum_{k\in K_h}\alpha_k\big)}{c_h} + \delta\cdot |K_h|$; \label{utility-bound-condition1}
\item for any student $i$ with $z_{h,i}<1$, $u_{h,i}\leq u_h^*$.\label{utility-bound-condition2}
\end{enumerate}
\end{claim}

\begin{proof}
The proof of Claim \ref{claim:bound_utility} is similar to the proof of Claim \ref{clm:prices_LEO}. To show (\ref{utility-bound-condition1}), we assume instead $u^*_h> \frac{\big(\sum_{k\in K_h}\alpha_k\big)+|K_h|\cdot c_h\cdot \delta}{c_h}$. Then, by the definition of $u_h^*$, it holds that $u_{h,j}>\frac{\big(\sum_{k\in K_h}\alpha_k\big)+|K_h|\cdot c_h\cdot \delta}{c_h}$ for any $j$ with $z_{h,j}>0$. As a result, the total utility from $\z_h$ becomes: 
$$U(\z_h) \ge \sum_{i \in S, z_i>0} u_h^* z_i \ge u_h^* \cdot c_h > \frac{\big(\sum_{k\in K_h}\alpha_k\big)+|K_h|\cdot c_h\cdot \delta}{c_h} \cdot c_h = \big(\sum_{k\in K_h}\alpha_k\big)+|K_h|\cdot c_h\cdot \delta,$$
which contradicts Lemma \ref{lemma:revenue_with_delta}.

To see (\ref{utility-bound-condition2}), we assume by contradiction there exists a student $i$ with $z_{h,i}<1$ and $u_{h,i}>u_h^*.$ Let $i^*$ be a student with $z_{h,i^*}>0$ and $u_{h,i^*} = u_h^*.$ Then, we can \textit{strictly} improve upon the utility of $\z_h$ by moving mass from $z_{h,i^*}$ to $z_{h,i}$, contradicting to the requirement that $\z_h$ is utility-maximizing.
\end{proof}

We now show that a MEO with $\delta>0$ and $\randbud{}=\mathbf{U}[1-\varepsilon,1]$ gives an approximately stable fractional matching in which individual rationality and no blocking are violated by small amounts. 
\begin{claim}\label{clm:MEO_approx}
Let $(\Hos, S, \{c_h\}_{h\in \Hos},\Vec{\alpha}:=\{\alpha_{k}\}_{h\in \Hos, k\in K_h})$ be an instance of the matching problem with $\sum_{h\in \Hos}c_h\leq |S|$. Let $(\x,\y,\z,\prices, \bold{q})$ be a MEO under $\delta>0$ and $\randbud{}=\mathbf{U}[1-\varepsilon, 1]$ with $\varepsilon>0$. Then, there exists a set of parameters $\{\beta_{h}\}_{h\in \Hos}$ with $\beta_h\leq \frac{\sum_{k\in K_h}\alpha_k}{c_h}+\delta\cdot |K_h|$ satisfying: 
\begin{enumerate}
\item $||\z_h||_1 = c_h$ for each school $h\in \Hos$; \label{MEO-approx-condition1}
\item $ \text{w-}\mathrm{Support}^{\vec{\alpha}}_h(i, \z) \ge \lceil \beta_h - \delta\cdot |K_h|\rceil $ for each school $h$ and student $i$ with $z_{h,i}>0$;\label{MEO-approx-condition2}
\item $\text{s-}\mathrm{Support}^{\vec{\alpha}}_h(i, \z) \le \lfloor\frac{\beta_h}{(1-\varepsilon)^2}\rfloor$ for each school $h$ and student $j$ with $z_{h,j}<1$ and $h\succ_j h_{j,\z}$. \label{MEO-approx-condition3}
\end{enumerate}
\end{claim}
\begin{proof}
For each school $h$, we set $\beta_h := u_h^* =\min_{i\in S}\{u_{h,i}:z_{h,i}>0\}$ from Claim \ref{claim:bound_utility}. 


Then, we can see that condition (\ref{MEO-approx-condition1}) is trivially satisfied by the definition of MEO. We are now left to show conditions (\ref{MEO-approx-condition2}) and (\ref{MEO-approx-condition3}).

To prove (\ref{MEO-approx-condition2}), we pick an arbitrary student $i$ with $z_{h,i}>0$ and a school $h$. We use $K_{A}$ to denote the set of committee members who weakly approve $i$ and $K_{NA}$ to denote those who do not weakly approve $i$. We can use the same argument as in Lemma \ref{clm:LEO_approx} to show that the personalized price of $i$ for any committee $k\in K_{NA}$ is at most $\delta$. Then, with $q_{i,h}\leq 1$, the utility $q_{i,h}\cdot p_{k,i}$ of $i$ collected from $k\in K_{NA}$ is at most $\delta$. Furthermore, for any committee member $k\in K_{A}$, the utility $q_{i,h}\cdot p_{k,i}$ is at most $1$. Then, combining the discussion above with  $u_{h,i}\geq u_h^*$, we can derive 
$$|K_A| + \delta\cdot |K_{NA}|\geq \sum_{k \in K_A} q_{i,h}\cdot p_{k,i} + \sum_{k\in K_{NA}} q_{i,h}\cdot p_{k,i} = u_{h,i}\geq u_h^*,$$
from which we obtain $\text{w-}\mathrm{Support}^{\vec{\alpha}}_h(i, \z)=|K_A|\geq \lceil u_h^* - \delta\cdot |K_{NA}| \rceil \geq \lceil u_h^*-\delta\cdot |K_h|\rceil $. 


We now prove (\ref{MEO-approx-condition3}). Pick a student $j$ with $z_{h,j}<1$ and a school $h\succ_j h_{j,\z}$. By our characterization of boundary school in Lemma \ref{lemma:boundary_school_price}, we have $q_{j,h}>1-\varepsilon$. Then, we use $K_{h,j}$ to denote the set of voters who strongly approves $j$ with respect to $\z$. By the characterization of $\alpha_k$-rank student in Lemma \ref{lemma:boundary_student}, we know that $p_{k,j}>1-\varepsilon$ for any $k\in K_{h,j}$. Then, combining the discussion above and the result of $u_{h,j}\leq u_h^*$ as given Claim \ref{claim:bound_utility}, we can derive
$$(1-\varepsilon)^2\cdot |K_{h,j}|=\sum_{k\in K_{h,j}}(1-\varepsilon)^2\leq \sum_{k\in K_{h,j}}q_{j,h}\cdot p_{k,j}\leq \sum_{k\in K_h}q_{j,h}\cdot p_{k,j} = u_{h,j}\leq u_h^*,$$
from which we obtain $|K_{h,j}|\leq \lfloor\frac{u_h^*}{(1-\varepsilon)^2}\rfloor$ as desired. 
\end{proof}

The existence of stable fractional matching then follows from Claim \ref{clm:MEO_approx} by picking small enough $\varepsilon$ and $\delta$. 

\Xomit{
\begin{theorem}\label{thm:matching_exists}
Let $(\Hos, S, \{c_h\}_{h\in \Hos},\Vec{\alpha}:=\{\alpha_{k}\}_{k\in K})$ be an instance of the matching problem with $\sum_{h\in \Hos}c_h\leq |S|$. Then, there exist $\varepsilon, \delta>0$ s.t. the fractional matching $\z:=\{\z_h\}_{h\in \Hos}$ of the corresponding MEO is stable with respect to parameters $\{\beta_h\}_{h\in \Hos}$ satisfying $\beta_h\leq \lceil\frac{\sum_{k\in K_h}\alpha_k}{c_h }\rceil$.
\end{theorem}
}
\begin{proof}[Proof of Theorem~\ref{thm:matching_exists}]

Finally, Theorem~\ref{thm:matching_exists} follows from Claim~\ref{clm:MEO_approx} by choosing $\varepsilon,\delta>0$ sufficiently small. Let $M=\max_{h\in \Hos} |K_h|$. Then, we can pick $\varepsilon>0$ and $\delta>0$ so that,  
\[
\frac{M}{(1-\varepsilon)^2} - M \leq \frac{1}{10} \quad \text{and} \quad 
\delta\leq \frac{1}{10M} \quad   \text{and} \quad \frac{1}{(1-\varepsilon)^2}\leq 2.
\]  
Let $\{\beta_h\}_{h\in \Hos}$ be the set of parameters in Claim~\ref{clm:MEO_approx}, and let $(\x,\y,\z,\prices,\bold{q})$ be the corresponding MEO. Next, we define  
\[
\beta^*_h := \left\lfloor \frac{\beta_h}{(1-\varepsilon)^2} \right\rfloor.
\]  
We can compute, with our choices of parameters, that for each $h\in \Hos$,  
\begin{align*}
 \frac{\beta_h}{(1-\varepsilon)^2} - \frac{\sum_{k\in K_h}\alpha_k}{c_h }& \leq \frac{1}{(1-\varepsilon)^2}\cdot \big( \frac{\sum_{k\in K_h}\alpha_k}{c_h } + \delta\cdot |K_h|  \big)- \frac{\sum_{k\in K_h}\alpha_k}{c_h }   \\
&\leq \frac{\sum_{k\in K_h}\alpha_k}{c_h }\cdot (\frac{1}{(1-\varepsilon)^2}-1) + \frac{\delta\cdot |K_h|}{(1-\varepsilon)^2}\\
&\leq M\cdot (\frac{1}{(1-\varepsilon)^2}-1)+ \frac{\delta\cdot |K_h|}{(1-\varepsilon)^2}\\
&<1. 
 \end{align*}
It then follows $\beta^*_h = \left\lfloor \frac{\beta_h}{(1-\varepsilon)^2} \right\rfloor\leq \lceil\frac{\sum_{k\in K_h}\alpha_k}{c_h }\rceil$ as desired, and we are left to verify that $\z$ is stable with respect to $\beta^*$. 

\begin{itemize}
    \item Condition~(\ref{MEO-approx-condition1}) of Claim~\ref{clm:MEO_approx} guarantees that each school is always of full capacity. 
    \item Condition~(\ref{MEO-approx-condition2}) of Claim~\ref{clm:MEO_approx} guarantees individual rationality. For each $i$ with $z_{h,i}>0$, we have 
    $ \text{w-}\mathrm{Support}^{\vec{\alpha}}_h(i, \z) \ge \lceil \beta_h - \delta\cdot |K_h|\rceil $. Then, by our choice of $\delta$, we can compute 
    \begin{align*}
    \frac{\beta_h}{(1-\varepsilon)^2} - (\beta_h - \delta\cdot |K_h|)& \leq \beta_h\cdot (\frac{1}{(1-\varepsilon)^2}-1) +\delta\cdot |K_h| \\
    &\leq \big( \frac{\sum_{k\in K_h}\alpha_k}{c_h } + \delta\cdot |K_h|  \big)\cdot (\frac{1}{(1-\varepsilon)^2}-1) + \delta\cdot |K_h| \\
    &\leq \frac{\sum_{k\in K_h}\alpha_k}{c_h }\cdot (\frac{1}{(1-\varepsilon)^2}-1) + \frac{\delta\cdot |K_h|}{(1-\varepsilon)^2}\\
    &\leq M\cdot (\frac{1}{(1-\varepsilon)^2}-1)+ \frac{\delta\cdot |K_h|}{(1-\varepsilon)^2}\\
&<1. 
    \end{align*}
This implies $\lceil \beta_h-\delta\cdot |K_h|\rceil \geq \lfloor \frac{\beta_h}{(1-\varepsilon)^2}\rfloor$ and therefore 
  $ \text{w-}\mathrm{Support}^{\vec{\alpha}}_h(i, \z) \ge \beta_h^*$ as desired. 
  \item No blocking follows directly from Condition~(\ref{MEO-approx-condition3}) of Claim~\ref{clm:MEO_approx}.
\end{itemize}
\end{proof}

\section{Rounding Lemmas for Stable Matching}\label{app:roundmatch}

\subsection{Proof of Lemma~\ref{lemma:round_matching}}
\begin{proof}
We start by observing that Claims~\ref{cor:support_inside} and~\ref{cor:support_outside} still hold true in this setting. Then, individual rationality follows from Claim~\ref{cor:support_inside} as in the proof of Lemma \ref{lemma:round}. 

Then, we are left to show no blocking for $\z$ with respect to  $\{\alpha'_h\}_{h\in \Hos}$ and $\{\beta_h\}_{h\in \Hos}$. In particular, it suffices to show that if a school-student pair $(h,i)$ with $i\notin M_h$ blocks $M$, then $(h,i)$ must block $\z$. By the definition of blocking, we have $\mathrm{Support}^{\vec{\alpha}'}_h(i, M_h)>\beta_h$. Then, it follows from Claim \ref{cor:support_outside} that $\text{s-}\mathrm{Support}^{\vec{\alpha}}_h(i, \z)>\beta_h$. We then break into two cases based on $M_i$ and show that $h\succ_i h_{i,\z}$ both cases. 
\begin{itemize}
\item [\textbf{Case 1}:]$M_i = \emptyset$. This implies that the amount of weight with which $i$ is mapped to $\emptyset$ in $\z$ is \textit{strictly} positive, and therefore $h\succ_i h_{i,\z}$. 
\item [\textbf{Case 2}:] $M_i\neq \emptyset$. According to the assumption of this Lemma, we have $z_{i,M_i}>0$. Since $h_{i,\z}$ is by definition the least favorite school of $i$ with a strictly positive weight, it must hold $M_i\succeq_i h_{i,\z}$. It then follows from transitivity that $h\succ_i h_{i,\z}$.  
\end{itemize}
Consequently, given that there is no blocking pair for $\z$, there cannot be any blocking pair for $M$ and therefore $M$ is stable. 
\end{proof}

\subsection{Proof of Lemma \ref{lemma:round_matching_violation}}
\label{app:round_matching_violation}
\begin{proof}
We round the fractional solution $\{\z_h\}_{h\in\Hos}$ to an integral vector $\{\z'_h\}_{h\in\Hos}$ using iterative rounding. 
First, we fix all variables $z_{h,i}\in\{0,1\}$ by setting $z'_{h,i}:=z_{h,i}$ and restrict attention to the remaining variables. Hence, without loss of generality, we assume $0<z_{h,i}<1$ for all remaining $i\in S$ and $h\in\Hos$.

For each school $h\in\Hos$ and committee member $k\in K_h$, define
\[
a_k := \Big\lfloor \sum_{i\in U_k(\z,\vec{\alpha})} z_{h,i} \Big\rfloor,
\qquad
b_k := \Big\lceil \sum_{i\in U_k(\z,\vec{\alpha})} z_{h,i} \Big\rceil .
\]
Consider the following linear program:
\begin{align*}
\sum_{h\in\Hos} x_{h,i} &= 1, 
&& \forall i\in S \text{ with } \sum_{h\in\Hos} z_{h,i}=1 
&& \text{(student demand)},\\
\sum_{i\in S} x_{h,i} &= c_h, 
&& \forall h\in\Hos 
&& \text{(school capacity)},\\
a_k \le \sum_{i\in U_k(\z,\vec{\alpha})} x_{h,i} &\le b_k,
&& \forall h\in\Hos,\ k\in K_h 
&& \text{(committee demand)},\\
0 \le x_{h,i} &\le 1,
&& \forall h\in\Hos,\ i\in S 
&& \text{(bounds)}.
\end{align*}

At each step, we apply the standard iterative rounding procedure
\citep{book11iterative,allocation2021,nguyen2016assignment} to obtain an integral vector $\{\z'_h\}_{h\in\Hos}$. The resulting solution satisfies all student demand constraints exactly and violates each school capacity by at most $2|K_h|+1$ and committee demand constraint by at most $2|K_h|+2$. We briefly sketch the algorithm for completeness.

\begin{itemize}
\item Find an extreme point $\{\x_h\}_{h\in\Hos}$ of the current LP with at least one integral coordinate.
\item For each integral coordinate $x_{h,i}=0$, set $z'_{h,i}=0$ and remove the variable $x_{h,i}$ from all constraints.
\item For each integral coordinate $x_{h,i}=1$, do the following:
\begin{itemize}
\item set $z'_{h,i}=1$;
\item remove the student demand constraint corresponding to student $i$;
\item update the school capacity constraint for school $h$ and all committee demand constraints containing $x_{h,i}$;
\item remove the variable $x_{h,i}$ from all constraints.
\end{itemize}
\item If a school capacity constraint for some $h$ contains at most $2|K_h|+1$ remaining variables, delete the constraint.
\item If a committee demand constraint for some $(h,k)$ contains at most $2|K_h|+2$ remaining variables, delete the constraint.
\item Repeat until either (i) no variables remain, or (ii) no school capacity or committee demand constraints remain.
\end{itemize}

If the algorithm terminates, the resulting vector $\{\z'_h\}_{h\in\Hos}$ satisfies the requirements of Lemmas~\ref{lemma:round_matching_violation} and~\ref{lemma:round_matching}. Thus, it remains to show that the algorithm indeed terminates.

Assume, toward a contradiction, that the algorithm becomes stuck. Then the current solution is an extreme point in which all variables are strictly fractional, i.e., $0<x_{h,i}<1$ for every remaining variable, and no constraint can be deleted.

We prove that this situation cannot occur using a token-counting argument. The key observation is the following fact:

\begin{quote}
In an extreme point  of a in a linear program, the number of variables equals the number of linearly independent binding constraints.
\end{quote}

We show that, in our setting, the number of variables must strictly exceed the number of active constraints, yielding a contradiction.

Assign one token to each variable $x_{h,i}$ and distribute it as follows:
\begin{itemize}
\item give $\tfrac12$ token to the student demand constraint corresponding to student $i$;
\item give $\tfrac{1}{2(|K_h|+1)}$ token to the school capacity constraint corresponding to school $h$;
\item give $\tfrac{1}{2(|K_h|+1)}$ token to each committee demand constraint for $k\in K_h$ such that $i\in U_k(\z,\vec{\alpha})$.
\end{itemize}
At each iteration, every active student demand constraint receives at least one token unless only a single variable remains. In the latter case, that variable must be integral, and the algorithm can proceed.

Now fix a school $h$. If the school capacity constraint for $h$ is active, then it must contain at least $2(|K_h|+1)$ remaining variables and therefore receives at least one token. Similarly, for each active committee demand constraint associated with $h$, the constraint must contain more than $2(|K_h|+1)$ remaining variables and therefore receives more than one token. Hence, if any committee demand constraint remains active, the total number of variables strictly exceeds the number of active constraints, contradicting the assumption that the current solution is an extreme point.

Finally, when all committee demand constraints have been removed, the remaining LP corresponds to a matching polytope, whose extreme points are integral. This completes the proof.

\end{proof}

\end{document}